\documentclass[12pt,letterpaper,onecolumn,draftcls]{IEEEtran}

\usepackage{cite,graphicx,epsfig,subfigure,amsmath,amssymb,mathrsfs,epsf}

\newtheorem{theorem}{Theorem}
\newtheorem{lemma}{Lemma}
\newtheorem{remark}{Remark}
\newtheorem{definition}{Definition}
\newtheorem{example}{Example}

\def \tr{\operatorname{tr}}

\begin{document}

\title{Cooperative Transmission for a Vector Gaussian Parallel Relay Network}
\author{Muryong~Kim,~\IEEEmembership{Member,~IEEE},
        and~Sae-Young~Chung,~\IEEEmembership{Senior Member,~IEEE}\\
\thanks{This work was supported in part by the Center for Broadband OFDM Mobile Access (BrOMA) through the ITRC Program of the Ministry of Information and Communication (MIC), Korea, supervised by the Institute of Information Technology Assessment (IITA). The material in this paper was presented in part at the 45th Annual Allerton Conference on Communication, Control, and Computing, Monticello, IL, September 2007.}
\thanks{The authors are with the School of Electrical Engineering and Computer Science, KAIST, Daejeon, Korea (e-mail: muryong@kaist.ac.kr, sychung@ee.kaist.ac.kr).}}


\maketitle

\begin{abstract}
In this paper, we consider a parallel relay network where two relays
cooperatively help a source transmit to a destination. We assume
the source and the destination
nodes are equipped with multiple antennas. Three basic schemes
and their achievable rates are studied: Decode-and-Forward (DF),
Amplify-and-Forward (AF), and Compress-and-Forward (CF). For the DF
scheme, the source transmits two private signals, one for each
relay, where dirty paper coding (DPC) is used between the two private
streams, and a common signal for both relays. The relays make
efficient use of the common information to introduce a proper amount
of correlation in the transmission to the destination. We show that
the DF scheme achieves the capacity under certain conditions. We
also show that the CF scheme is asymptotically optimal in the high
relay power limit, regardless of channel ranks. It turns out that
the AF scheme also achieves the asymptotic optimality but only when
the relays-to-destination channel is full rank. The relative
advantages of the three schemes are discussed with numerical
results.
\end{abstract}


\IEEEpeerreviewmaketitle

\section{Introduction}
\PARstart{O}{v}er the recent years, relaying has been considered as
a promising technique that can increase throughput and reliability
and enhance the coverage of wireless networks. There have been a
number of research results showing different aspects of relay
channels. Cover and El Gamal \cite{CoverElGamal:79} derived the
capacity of a class of relay channels with a single relay that helps
transmission from a source to a destination. Kramer et al.
\cite{KramerGastparGupta:2005} generalized the results in various
ways. Laneman et al. \cite{LanemanTseWornell:2004} considered
cooperative diversity aspects of relay channels. In
\cite{FanWangPoorThompson}, the authors consider the multiplexing
aspects of cooperative communications in multi-antenna relay
networks. In
\cite{WangZhangHostMadsen:2005,BolcskeiNabarOymanPaulraj:2006},
multi-antenna relay channels are considered. The use of relays in
broadcast scenarios is considered in
\cite{KramerGastparGupta:2005,LiangVeeravalli}.

In this paper, we study the capacity of a vector Gaussian parallel
relay network where two parallel relays help transmission from a
source node to a destination node. The network model is first
studied by Schein et al. \cite{ScheinThesis}, and a set of capacity
theorems are derived for the discrete memoryless channel and the
scalar Gaussian channel. The authors of \cite{XueSandhu} considered
a similar model but with half-duplex constraint, i.e., relays do not
transmit and receive at the same time. Recently, the authors of
\cite{AvestimehrDiggaviTse} showed their new achievable rate for
general scalar Gaussian relay networks is within a constant number
of bits from the cut-set upper bound on the capacity. A new
achievable rate for the original Schein's network is derived using a
Combined Amplify-and-Decode Forward (CADF) scheme in
\cite{RezaeiGharanKhandani}. Our network model is different from the
earlier ones in that both the broadcast channel (BC) part and the
multiple access channel (MAC) part are vector Gaussian channels as
the source and the destination nodes are equipped with multiple
antennas. As the vector BC is not degraded in general and a simple
superposition coding will not suffice. In the vector MAC,
correlation between relay signals are not always beneficial, rather,
the right amount of correlation may result in a better performance
as will be seen in a later section. Throughout the paper, upper
bounds and achievable rates by different cooperative transmission
strategies: DF, AF and CF are derived.

For the DF scheme, the vector Gaussian parallel relay network can be
seen as a cascade of multiple-input single-output (MISO) BC and
single-input multiple-output (SIMO) MAC channels. We first extend
some earlier results for the discrete memoryless and scalar MACs with
common information to the vector Gaussian case, and investigate the
characteristics of the three-dimensional achievable rate region. We
use a known transmission scheme of \cite{JindalGoldsmith:ISIT04} for
the BC part. Using the BC-MAC schemes, we show that DF achieves the
capacity of the vector Gaussian parallel relay network under certain
conditions. In addition, we address the importance of common
information signaling and correlation control.

We also extend some earlier results for AF and CF to the
vector Gaussian case. We show when DF is strictly
suboptimal and when AF and CF can outperform DF by comparing the
achievable rates and the upper bound. We show that AF is
asymptotically optimal in the high relay power limit if the channel
rank of the MAC part is full. In addition, we also show that the
CF scheme achieves the asymptotic capacity, regardless of the
channel ranks.

The rest of the paper is organized as follows. In Section II, we
introduce the system model. We derive a capacity upper bound for the
vector Gaussian parallel relay network in Section III. Next, we
derive achievable rates by the DF, AF and CF schemes in Sections IV,
V and VI, respectively. Numerical results and comparison of
different schemes are given in Section VII. Conclusions and final
remarks are given in Section VIII.

\section{System Model}
The vector Gaussian parallel relay network consists of four nodes: a
source, a destination, and two relays. We assume no direct link from
the source to the destination. The relays are assumed to be full
duplex, i.e., they transmit and receive at the same time. The
received signals at the relays and at the destination are given by
\begin{eqnarray}
y_{r1} &=& {\bf g}_1{\bf x}_s + n_{r1},\nonumber\\
y_{r2} &=& {\bf g}_2{\bf x}_s + n_{r2},\nonumber\\
{\bf y}_d &=& {\bf h}_1 x_{r1}  + {\bf h}_2 x_{r2}+ {\bf
n}_d,\nonumber
\end{eqnarray}
where
\begin{itemize}
  \item ${\bf x}_s \in { \mathbb{C} }^{M \times {\rm 1}}, x_{r1},
x_{r2}\in \mathbb{C}$ are the transmitted signals from the source
and from relays 1 and 2, respectively. Input covariance matrix and
power constraint at the source node are given by
 ${\bf Q}_s = {\mathbb E}[ {{\bf
x}_s{\bf x}_s^\dagger  }]$ and ${\rm tr}\left( {{\bf Q}_s } \right)
\leq P_s$, respectively. Power constraints at relays are given by
${\mathbb E}[|x_{r1}|^2] \leq P_{r1}$ and ${\mathbb E}[|x_{r2}|^2]
\leq P_{r2}$;
  \item $y_{r1}, y_{r2}  \in \mathbb{C}, {\bf y}_d \in {
\mathbb{C} }^{N \times 1}$ are the received signals at relays 1 and
2 and at the destination, respectively;
  \item ${\bf g}_1, {\bf g}_2 \in { \mathbb{C} }^{{\rm 1} \times
M}$ are the channel gains from the source to relays 1 and 2,
respectively, and ${\bf h}_1, {\bf h}_2 \in { \mathbb{C} }^{N \times
1}$ are the channel gains from relays 1 and 2 to the destination,
respectively;
  \item $n_{r1}, n_{r2} \in \mathbb{C}, {\bf n}_d \in { \mathbb{C}
}^{N \times 1}$ are additive white Gaussian noise (AWGN) at relays 1
and 2 and at the destination, respectively. Noise at the relays and
each antenna of the destination node is circularly symmetric complex
Gaussian, i.e., $n_{r1} ,n_{r2} \sim \mathcal{CN}\left( {0,1}
\right)$ and ${\bf n}_d \sim \mathcal{CN}\left( 0,
\mathbf{I}\right)$ and they are all independent of each other and
from the signals.
\end{itemize}
Throughout the paper, the following notation will be used,
\begin{itemize}
  \item The vector and matrix notations:
\[
\mathbf{y}_r  = \left[ {\begin{array}{*{20}c}
   {y_{r1} }  \\
   {y_{r2} }  \\
 \end{array} } \right]
, \mathbf{x}_r  = \left[ {\begin{array}{*{20}c}
   {x_{r1} }  \\
   {x_{r2} }  \\
 \end{array} } \right]
, \mathbf{n}_r  = \left[ {\begin{array}{*{20}c}
   {n_{r1} }  \\
   {n_{r2} }  \\
 \end{array} } \right]
, \mathbf{G} = \left[ {\begin{array}{*{20}c}
   {\mathbf{g}_1 }  \\
   {\mathbf{g}_2 }  \\
 \end{array} } \right]
, \mathbf{H} = \left[ {\begin{array}{*{20}c}
   {\mathbf{h}_1 } & {\mathbf{h}_2 }  \\
 \end{array} } \right].
\]
  \item ${\bf x}_s^n$, $x_{r1}^n$, $x_{r2}^n$,
$y_{r1}^n$, $y_{r2}^n$, ${\bf y}_d^n$ denote length-$n$ sequences of
${\bf x}_s$, $x_{r1}$, $x_{r2}$, $y_{r1}$, $y_{r2}$, ${\bf y}_d$,
respectively.
  \item $\mathbb{E}_{\bf x}[\cdot]$ denotes expectation with
respect to the distribution of $\bf x$, $\mathbb{E}_x[\cdot|y]$ does
expectation with respect to the distribution of $x$ conditioned on
$y$, and the simpler notation $\mathbb{E}[\cdot]$ without subscript
will be used as long as it is apparent.
  \item $(\cdot)^{opt}$ means the optimal value of a variable.
\end{itemize}

An example for $M=N=2$ case where $\mathbf H$ and $\mathbf G$ are
two-input two-output channels is shown in Fig. \ref{f:systemmodel}.

\begin{definition} A $(2^{nR},n)$ code for vector Gaussian
parallel relay network consists of a message set $W =
\{1,2,...,2^{nR}\}$, an encoding function at the source $f_s: w \in
\{1,2,...,2^{nR}\}\rightarrow \mathbb{C}^{M\times n}$, relaying functions at
two relays, $f_{r1,i}: \mathbb{C}^{i-1} \rightarrow \mathbb{C}$ and $f_{r2,i}:
\mathbb{C}^{i-1} \rightarrow \mathbb{C}$, respectively,
where $1\leq i \leq n$ is the time index\footnote{This means that
the output of a relay at time $i$ depends on all past received symbols.} and a decoding
function at the destination $g_d: \mathbb{C}^{N\times n} \rightarrow \hat w \in
\{1,2,...,2^{nR}\}$.
If the message
$w \in W$ is sent, the conditional probability of error is defined as $\lambda
(w) = \Pr\{g_d({\bf y}_d^n)\neq w|w \textrm{ sent}\}$. The average
probability of error is defined as
$P_e^{(n)}=\frac{1}{2^{nR}}\sum\limits_w \lambda (w)$.
\end{definition}

\begin{definition} If there exists a sequence of $(2^{nR},n)$
codes with $P_e^{(n)}\rightarrow 0$, the rate $R$ is said to be
achievable.
\end{definition}

\begin{definition} The capacity $C$ of a vector Gaussian parallel
relay network is the supremum of the set of achievable rates.
\end{definition}

\section{Capacity Upper Bound}
In this section, we derive the
cut-set bound~\cite{CoverThomas} applied
to the vector Gaussian parallel relay network. From the four cuts shown in Fig.
\ref{f:cut-set}, we get the following cut-set bound:
\[
\max\limits_{p(\mathbf{x}_s,x_{r1},x_{r2})}\min \left[
I(\mathbf{x}_s ;y_{r1} ,y_{r2},\mathbf{y}_d|x_{r1},x_{r2} ), I(x_{r1} ,x_{r2} ;\mathbf{y}_d ),
I(\mathbf{x}_s,x_{r1} ;\mathbf{y}_d,y_{r2}|x_{r2} ),
I(\mathbf{x}_s,x_{r2} ;\mathbf{y}_d,y_{r1}|x_{r1} )\right].
\]
Using the Markovity of the channel, i.e.,
$\mathbf{x}_s\leftrightarrow (x_{r1},x_{r2})\leftrightarrow \mathbf{y}_d$ and $(x_{r1},x_{r2},\mathbf{y}_d) \leftrightarrow \mathbf{x}_s \leftrightarrow (y_{r1},y_{r2})$,
it is easy to get the following loosened cut-set bound:
\[
\max\limits_{p(\mathbf{x}_s,x_{r1},x_{r2})}\min \left[
I(\mathbf{x}_s ;y_{r1} ,y_{r2} ), I(x_{r1} ,x_{r2} ;\mathbf{y}_d ),
I(\mathbf{x}_s ;y_{r2} ) + I(x_{r1} ;\mathbf{y}_d |x_{r2} ),
I(\mathbf{x}_s ;y_{r1} ) + I(x_{r2} ;\mathbf{y}_d |x_{r1} )\right].
\]
For example, for the first term we get
\begin{align*}
I(\mathbf{x}_s ;y_{r1} ,y_{r2},\mathbf{y}_d|x_{r1},x_{r2} ) &=
I(\mathbf{x}_s ;y_{r1} ,y_{r2}|x_{r1},x_{r2} )+I(\mathbf{x}_s ;\mathbf{y}_d|x_{r1},x_{r2},y_{r1} ,y_{r2})\\
&\leq I(\mathbf{x}_s ;y_{r1} ,y_{r2})+I(\mathbf{x}_s ;\mathbf{y}_d|x_{r1},x_{r2},y_{r1} ,y_{r2})\\
&=I(\mathbf{x}_s ;y_{r1} ,y_{r2})
\end{align*}
where the inequality and the last equality follows from the Markovity $(x_{r1},x_{r2},\mathbf{y}_d)\leftrightarrow \mathbf{x}_s \leftrightarrow (y_{r1},y_{r2})$.
For the third term we get
\begin{align*}
I(\mathbf{x}_s,x_{r1} ;\mathbf{y}_d,y_{r2}|x_{r2} ) &=
I(\mathbf{x}_s,x_{r1};\mathbf{y}_d|x_{r2})+I(\mathbf{x}_s,x_{r1} ;y_{r2}|x_{r2}, \mathbf{y}_d) \\
&=I(x_{r1};\mathbf{y}_d|x_{r2})+I(\mathbf{x}_s;\mathbf{y}_d|x_{r1},x_{r2})+I(\mathbf{x}_s,x_{r1} ;y_{r2}|x_{r2},\mathbf{y}_d ) \\
&=I(x_{r1};\mathbf{y}_d|x_{r2})+I(\mathbf{x}_s,x_{r1} ;y_{r2}|x_{r2},\mathbf{y}_d ) \\
&=I(x_{r1};\mathbf{y}_d|x_{r2})+I(\mathbf{x}_s ;y_{r2}|x_{r2},\mathbf{y}_d )+I(x_{r1} ;y_{r2}|\mathbf{x}_s,x_{r2},\mathbf{y}_d ) \\
&=I(x_{r1};\mathbf{y}_d|x_{r2})+I(\mathbf{x}_s ;y_{r2}|x_{r2},\mathbf{y}_d )\\
&\leq I(x_{r1};\mathbf{y}_d|x_{r2})+I(\mathbf{x}_s ;y_{r2})
\end{align*}
where the third equality follows from the Markovity $\mathbf{x}_s\leftrightarrow (x_{r1},x_{r2})\leftrightarrow \mathbf{y}_d$ and the fifth equality and the inequality follow from the Markovity $(x_{r1},x_{r2},\mathbf{y}_d) \leftrightarrow \mathbf{x}_s \leftrightarrow (y_{r1},y_{r2})$.

Note that this cut-set bound is optimized over the joint input
distribution $p(\mathbf{x}_s,x_{r1},x_{r2})$.
Using this, we get the following capacity upper bound for the vector Gaussian parallel relay network.

\begin{theorem}
The capacity of the vector Gaussian parallel relay network
is upper bounded by the minimum of the three expressions given by
\begin{eqnarray}
C \leq R_{sum,BC}^{upper} = \mathop {\max }\limits_{\mathbf{Q}_s,
\tr (\mathbf{Q}_s) \leq P_s} \log \det \left( {\mathbf{I}  +
\mathbf{GQ}_s \mathbf{G}^\dag } \right),
\label{eqn:broadcast_cut-set}
\end{eqnarray}
\begin{eqnarray}
C \leq \mathop {\max }\limits_{|\rho| \in [0,1]} \min \left[ {\log
\det \left( {\mathbf{I}  + \mathbf{HQ}_r \mathbf{H}^\dag  } \right),
\log [( {1 + P_s \left\| {\mathbf{g}_2 } \right\|^2 })( {1 +
(1-|\rho|^2) P_{r1} \left\| {\mathbf{h}_1 } \right\|^2  } )]}
\right], \label{eqn:correlation_bound1}
\end{eqnarray}
\begin{eqnarray}
C \leq \mathop {\max }\limits_{|\rho| \in [0,1]} \min \left[ {\log
\det \left( {\mathbf{I}  + \mathbf{HQ}_r \mathbf{H}^\dag  } \right),
\log[( {1 + P_s \left\| {\mathbf{g}_1 } \right\|^2 })( 1 +
(1-|\rho|^2) P_{r2} \left\| {\mathbf{h}_2 } \right\|^2  )]} \right].
\label{eqn:correlation_bound2}
\end{eqnarray}
where $\mathbf{Q}_r = \mathbb{E}[\mathbf{x}_r \mathbf{x}_r^\dag] =
\left[ {\begin{array}{*{20}c}
   {P_{r1} } & { |\rho| e^{j\angle \mathbf{h}_1 ^\dag
\mathbf{h}_2} {\sqrt {P_{r1} P_{r2} } } }  \\
   {|\rho| e^{-j\angle \mathbf{h}_1 ^\dag
\mathbf{h}_2} {\sqrt {P_{r1} P_{r2} } }}& {P_{r2} }  \\
\end{array} } \right]$.
\end{theorem}

$\\$
\begin{proof}
From the first cut, the following upper bound is derived:
\begin{eqnarray}
  I(\mathbf{x}_s ;y_{r1} ,y_{r2} ) &=& h(y_{r1} ,y_{r2} ) - h(y_{r1} ,y_{r2} |\mathbf{x}_s ) \nonumber\\
   &=& h(\mathbf{y}_r ) - h(\mathbf{n}_r ) \nonumber\\
   &\leq& \log (2\pi e)^2 \det \mathbb{E}[\mathbf{y}_r \mathbf{y}_r^{\dag } ] - \log (2\pi e)^2 \det \mathbb{E}[\mathbf{n}_r \mathbf{n}_r^{\dag } ] \nonumber\\
   &=& \log \det \mathbb{E}[\mathbf{n}_r \mathbf{n}_r^{\dag }  + \mathbf{Gx}_s \mathbf{x}_s^{\dag } \mathbf{G}^{\dag } ] \nonumber\\
   &=& \log \det (\mathbf{I + GQ}_s \mathbf{G}^{\dag } ) \nonumber
\end{eqnarray}
where the inequality follows from the fact that the circularly
symmetric complex Gaussian maximizes the
entropy~\cite{Telatar:1999}. From the second cut, the following
upper bound is derived:
\begin{eqnarray}
  I(x_{r1} ,x_{r2} ;\mathbf{y}_d ) &=& h(\mathbf{y}_d ) - h(\mathbf{y}_d |x_{r1} ,x_{r2} ) \nonumber \\
   &=& h(\mathbf{y}_d ) - h(\mathbf{n}_d ) \nonumber\\
   &\leq& \log (2\pi e)^N \det \mathbb{E}[\mathbf{y}_d \mathbf{y}_d^\mathbf{\dag } ] - \log (2\pi e)^N \det \mathbb{E}[\mathbf{n}_d \mathbf{n}_d^\mathbf{\dag } ] \nonumber\\
   &=& \log \det \mathbb{E}[\mathbf{n}_d \mathbf{n}_d^\mathbf{\dag }  + \mathbf{h}_1 \mathbf{h}_1^\mathbf{\dag } |x_{r1} |^2  + \mathbf{h}_2 \mathbf{h}_2^\mathbf{\dag } |x_{r2} |^2  + \mathbf{h}_1 \mathbf{h}_2^\mathbf{\dag } x_{r1} x_{r2}^*  + \mathbf{h}_2 \mathbf{h}_1^\mathbf{\dag } x_{r1}^* x_{r2} ] \nonumber\\
   &=& \log \det (\mathbf{I} + \mathbf{h}_\mathbf{1} \mathbf{h}_\mathbf{1}^\mathbf{\dag } P_{r1}  + \mathbf{h}_2 \mathbf{h}_2^\mathbf{\dag } P_{r2}  + \mathbf{h}_1 \mathbf{h}_2^\mathbf{\dag } \rho \sqrt{P_{r1} P_{r2}} + \mathbf{h}_2 \mathbf{h}_1^\mathbf{\dag } \rho ^* \sqrt{P_{r1} P_{r2}}) \nonumber\\
   &=& \log \det (\mathbf{I} + \mathbf{HQ}_r \mathbf{H}^\mathbf{\dag } ). \nonumber
\end{eqnarray}
where $\rho  = \frac{{\mathbb{E}[x_{r1} x_{r2}^* ]}} {{\sqrt {P_{r1}
P_{r2} } }} $ and note that we set $\angle \rho = \angle
\mathbf{h}_1 ^\dag \mathbf{h}_2$ in (\ref{eqn:correlation_bound1})
and (\ref{eqn:correlation_bound2}). From the third cut, the
following upper bounds are derived:
\begin{eqnarray}
I(\mathbf{x}_s ;y_{r2} ) + I(x_{r1} ;\mathbf{y}_d |x_{r2} ) &\leq&
\log (1 + \mathbf{g}_2 \mathbf{Q}_s \mathbf{g}_2^\dag ) + \log (1 +
(1-|\rho|^2) P_{r1} \left\| {\mathbf{h}_1 } \right\|^2 ) \nonumber \\
&\leq& \log (1 + P_s \left\|{\mathbf {g}_2}\right\|^2 ) + \log (1
+(1-|\rho|^2) P_{r1} \left\| {\mathbf{h}_1 } \right\|^2 ) \nonumber
\end{eqnarray}
where the first inequality follows from the fact that the circularly
symmetric complex Gaussian distribution maximized differential
entropy and from the following property:
\[
\mathbb{E}_{x_{r2}}\left[ \mathbb{E}_{x_{r1}} \{\left| {x_{r1} }
\right|^2 |x_{r2}\}-\left|\mathbb{E}_{x_{r1}} \{x_{r1}
|x_{r2}\}\right|^2 \right] \leq (1-|\rho|^2)  P_{r1},
\]
and for the second inequality $\mathbf {Q}_s = \frac{{P_s}}{{
\left\|\mathbf{g}_2\right\|^2}} {\mathbf {g}_2^\dag \mathbf{g}_2}$
is used. Similarly, from the fourth cut:
\begin{eqnarray}
I(\mathbf{x}_s ;y_{r1} ) + I(x_{r2} ;\mathbf{y}_d |x_{r1} ) &\leq&
\log (1 + \mathbf{g}_1 \mathbf{Q}_s \mathbf{g}_1^\dag ) + \log (1 +
(1-|\rho|^2) P_{r2} \left\| {\mathbf{h}_2 } \right\|^2 ) \nonumber \\
&\leq& \log (1 + P_s \left\|\mathbf {g}_1\right\|^2 ) + \log (1 +
(1-|\rho|^2) P_{r2} \left\| {\mathbf{h}_2 } \right\|^2 ). \nonumber
\end{eqnarray}
By properly combining the above four bounds, we can see that the
minimum of the three expressions given in the theorem statement
results in a tighter upper bound.
\end{proof}

\section{Decode-and-Forward (DF)}
In this section, we describe an achievable rate of the vector
Gaussian parallel relay network with a DF strategy. With DF, the signal
is delivered from the source to the destination in a two hop
transmission. The source has three message sets: $m_1 \in
\{1,2,...,2^{nR_1}\}$ intended for relay 1, $m_2 \in
\{1,2,...,2^{nR_2}\}$ for relay 2, and $m_c \in
\{1,2,...,2^{nR_c}\}$ for both relays. Depending on the messages to
transmit, the source node makes the signal $\mathbf{x}_s $ as a
function of $m_1$, $m_2$ and $m_c$. Upon successful decoding of the
signals, relays know their private messages and the common message.
The relays re-encode the received information to make input signals
$x_{r1}$ as a function of $m_1$ and $m_c$, and $x_{r2}$ as a
function of $m_2$ and $m_c$ as a block Markov manner. As the end-to-end channel is a cascade
of MISO BC and SIMO MAC both with common information, we first
investigate optimal signaling in the second hop.

\subsection{The Second Hop: SIMO MAC with Common Information}
The capacity region of the discrete memoryless MAC with common
information is derived in \cite{SlepianWolf,Willems:82} and that of
scalar Gaussian MAC with common information in
\cite{PrelovvanderMeulen:ISIT91}. The characteristics of the scalar
Gaussian MAC were further investigated in \cite{LiuUlukus:JCOM06}.
The result can be extended to the Gaussian MAC with multiple
antennas at the destination.

\begin{definition} A $((2^{nR_1}\times 2^{nR_c},2^{nR_2}\times
2^{nR_c}),n)$ code for the SIMO MAC with common information consists
of three message sets $m_1 = \{1,2,...,2^{nR_1}\}$, $m_2 =
\{1,2,...,2^{nR_2}\}$, $m_c = \{1,2,...,2^{nR_c}\}$, encoding
functions at two relays, $f_{r1}: (m_1,m_c) \rightarrow
\mathbb{C}^{n}$ and $f_{r2}: (m_2,m_c) \rightarrow \mathbb{C}^{n}$,
respectively, and a decoding function at the destination $g_d:
\mathbb{C}^{M\times n} \rightarrow (\hat m_1,\hat m_2,\hat m_c )$.
If the messages $(m_1,m_2,m_c)$ are sent, the conditional
probability of error is $\lambda (m_1,m_2,m_c) = \Pr\{g_d({\bf
y}_d^n)\neq (m_1,m_2,m_c)|(m_1,m_2,m_c) \textrm{ sent}\}$. The
average probability of error is defined as
$P_e^{(n)}=\frac{1}{2^{n(R_1+R_2+R_c)}}\sum\limits_{(m_1,m_2,m_c)}
\lambda (m_1,m_2,m_c)$.
\end{definition}

\begin{definition} If there exists a sequence of $((2^{nR_1}\times
2^{nR_c},2^{nR_2}\times 2^{nR_c}),n)$ codes with
$P_e^{(n)}\rightarrow 0$, the rate triplet $(R_1,R_2,R_c)$ is said
to be achievable.
\end{definition}

\begin{definition} The capacity region of the SIMO MAC with common
information is the closure of the set of all achievable rate triplets.
\end{definition}

We will derive an achievable rate region assuming Gaussian input
distributions. Each relay's input signal is a superposition of
private and common signals:
\[
x_{r1}^n(m_1,m_c) = v^n(m_1) + \sqrt {\frac{{\alpha  P_{r1} }}{{P_c
}}}u_1^n(m_c), \ \ \ \ \ \ x_{r2}^n(m_2,m_c) = w^n(m_2) +  \sqrt
{\frac{{ \beta P_{r2} }}{{P_c }}} u_2^n(m_c)
\] where $v \sim
\mathcal{CN}\left( {0,\overline \alpha P_{r1}} \right)$, $w \sim
\mathcal{CN}\left( {0,\overline \beta P_{r2}} \right)$,
$\overline{\alpha}=1-\alpha$, $\overline{\beta}=1-\beta$, $0 \leq
\alpha \leq 1$, $0 \leq \beta \leq 1$, and $u_1$ and $u_2$ are
partially correlated random variables in the sense that
\[
\mathbf u= \left[ {\begin{array}{*{20}c}
   {u_1}  \\
   {u_2}  \\
 \end{array} } \right]\sim \mathcal{CN}\left( {0,P_c\left[ {\begin{array}{*{20}c}
   {1} & \gamma  \\
   \gamma^* & {1}  \\
 \end{array} } \right]} \right).
\]
The conditional probability density functions are given by
\[
p(x_{r1}|u_1) \sim \mathcal{CN}\left( {\sqrt {\frac{{\alpha  P_{r1}
}}{{P_c }}}u_1,\overline \alpha P_{r1}} \right), \ \ \ \ \ \
p(x_{r2}|u_2)\sim \mathcal{CN}\left( {\sqrt {\frac{{ \beta P_{r2}
}}{{P_c }}} u_2,\overline \beta P_{r2}} \right).
\]
Note that the correlation coefficient between the relay input
signals is given by $\rho  = \frac{{\mathbb{E}[x_{r1} x_{r2}^* ]}}
{{\sqrt {P_{r1} P_{r2} } }}=\gamma {\sqrt {\alpha \beta }}$ and
controllable by power allocation. The received signal at the
destination is given by \setlength{\arraycolsep}{.1em}
\begin{eqnarray}
\begin{array}{l}
{\bf y}_d = {\bf h}_1 x_{r1}  + {\bf h}_2 x_{r2}  + {\bf n}_d \\
  = {\bf h}_1 \left( v + \sqrt {\frac{{\alpha P_{r1}
}}{{P_c }}}u_1 \right) + {\bf h}_2 \left( w + \sqrt
{\frac{{ \beta P_{r2} }}{{P_c }}} u_2 \right) + {\bf n}_d \\
  = {\bf h}_1 v + {\bf h}_2 w + \left( {\bf h}_1 \sqrt {\frac{{\alpha P_{r1}
}}{{P_c }}}u_1 + {\bf h}_2 \sqrt
{\frac{{ \beta P_{r2} }}{{P_c }}} u_2 \right) + {\bf n}_d. \\
  = {\bf h}_1 v + {\bf h}_2 w + \left[ {\begin{array}{*{20}c}
   {\mathbf{h}_1 \sqrt {\frac{{\alpha
P_{r1} }}{{P_c }}}} & {\ \ \ \mathbf{h}_2 \sqrt {\frac{{\beta
P_{r2} }}{{P_c }}}}  \\
 \end{array} } \right] \mathbf u + {\bf n}_d. \\
\end{array}
\end{eqnarray}

The rate region of the SIMO Gaussian MAC in terms of mutual
information expression can be written as follows
\[
\mathcal{R}_{MAC}= \left\{ {\left( {R_1 ,R_2 ,R_c } \right)\left|
{\begin{array}{*{20}c}
   {R_1  \leq I(x_{r1}; \mathbf{y}_d | x_{r2}, u_1, u_2)}, \hfill  \\
   {R_2  \leq I(x_{r2}; \mathbf{y}_d | x_{r1}, u_1, u_2)}, \hfill  \\
   {R_1  + R_2  \leq I(x_{r1}, x_{r2}; \mathbf{y}_d | u_1, u_2)}, \hfill  \\
   {R_1  + R_2  + R_c  \leq I(x_{r1}, x_{r2}; \mathbf{y}_d )} \hfill  \\
 \end{array} } \right.} \right\}
\]
for some distribution $p(u_1,u_2)p(x_{r1}|u_1)p(x_{r2}|u_2)$ where
$u_1$ and $u_2$ are auxiliary random variables that represent common
information.

\begin{theorem}
The following rate region is achievable for the SIMO MAC with
common information:
\begin{eqnarray}
\mathcal{R}_{MAC} \left( {P_{r1} ,P_{r2} } \right) =
{\bigcup\limits_{\alpha ,\beta, \mathbb{E}[|x_i|^2]\leq P_{ri},
i=1,2 } {\mathcal{R}\left( \alpha ,\beta, \gamma \right)} }
\end{eqnarray}
where ${\mathcal{R}\left( \alpha ,\beta, \gamma \right)}$ is the
rate region for given $\alpha$, $\beta$ and $\gamma$, which can be
expressed as \setlength{\arraycolsep}{.1em}
\begin{eqnarray}
\mathcal{R}\left( {\alpha ,\beta, \gamma } \right) = \left\{ {\left(
{R_1 ,R_2 ,R_c } \right)\left| {\begin{array}{*{20}c}
   {R_1  \leq \log ( {1 + \overline \alpha P_{r1} \left\| {\mathbf{h}_1 } \right\|^2 } )}, \hfill \\
   {R_2  \leq \log ( {1 + \overline \beta P_{r2} \left\| {\mathbf{h}_2 } \right\|^2 } )}, \hfill \\
   {R_1  + R_2  \leq \log \det \left( {\mathbf{I} + \mathbf{HQ}_r^p \mathbf{H}^{\dag } } \right)}, \hfill \\
   {R_1  + R_2  + R_c  \leq \log \det \left( {\mathbf{I} + \mathbf{HQ}_r \mathbf{H}^{\dag } } \right)} \hfill \\
\end{array} } \right.} \right\}
\label{eqn:SIMO_MAC_region}
\end{eqnarray}
where
\begin{eqnarray}
&&\mathbf{Q}_r^p={\mathbb E}_u \left[{\mathbb E}_{{\bf x}_r}[ {\bf
x}_r {\bf x}_r^\dag  |u]\right]  = \left[ {\begin{array}{*{20}c}
   {\overline \alpha P_{r1} } & 0  \\
   0 & {\overline \beta P_{r2} }  \\
 \end{array} } \right]\nonumber\\
&&\mathbf{Q}_r={\mathbb E}_{{\bf x}_r}[ {\bf x}_r {\bf x}_r^\dag ]
  = \left[ {\begin{array}{*{20}c}
   {P_{r1} } & \gamma {\sqrt {\alpha \beta  P_{r1} P_{r2} } }  \\
   \gamma^* {\sqrt { \alpha \beta  P_{r1} P_{r2} }} & {P_{r2} }  \\
 \end{array} } \right].\nonumber
\end{eqnarray}
\end{theorem}

\begin{proof}
It is straightforward to show the theorem result by evaluating the
mutual information expressions assuming circularly symmetric complex
Gaussian input distributions.
\end{proof}

Next, we are interested in how to maximize the sum-rate and get the
following result for the optimal correlation that maximizes $\log
\det \left( {\mathbf{I} + {\bf HQ}_r {\bf H}^\dagger} \right)$.

\begin{lemma}[Optimal correlation]\label{lemma:optimal_correlation}
\label{th:corr} For any $N$, the sum rate of the SIMO MAC with
common information is maximized by $\angle\rho=\angle\gamma=\angle
\mathbf{h}_1^\dagger \mathbf{h}_2$ and $|\rho|
=|\gamma|\sqrt{\alpha\beta}= \min \left( {\frac{{|\mathbf
{h}_1^\dagger \mathbf {h}_2|}}{{\sqrt {P_{r1} P_{r2}} \det
\left(\mathbf {H^\dagger H}\right)}},1} \right)$, and the resulting
maximum sum-rate is given by
\[
\begin{gathered}
  R_{sum,MAC}^{max} = \mathop {\max }\limits_\rho  \log \det (\mathbf{I + HQ}_r \mathbf{H}^\mathbf{\dag } ) \hfill \\
= \left\{ {\begin{array}{*{20}c}
   {\log \left(1 + P_{r1} \left\|
{\mathbf{h}_1 } \right\|^2 + P_{r2} \left\| {\mathbf{h}_2 }
\right\|^2  + P_{r1} P_{r2} \det (\mathbf{H}^{\dag } \mathbf{H}) +
\frac{{| {\mathbf{h}_1^\mathbf{\dag } \mathbf{h}_2 } |^2 }}
{{\det (\mathbf{H}^{\dag } \mathbf{H})}} \right),} & {|\rho|^{opt}<1}  \\
   {\log \left(1 + P_{r1} \left\| {\mathbf{h}_1 } \right\|^2 + P_{r2}
\left\| {\mathbf{h}_2 } \right\|^2  +  2\sqrt {P_{r1} P_{r2} } |
{\mathbf{h}_1^{\dag } \mathbf{h}_2 } |\right),} & {|\rho|^{opt}=1}.  \\
 \end{array} } \right.
\end{gathered}
\]
\end{lemma}
$\ $\\
\begin{proof} See Appendix.
\end{proof}

\begin{remark} By defining the angle between channel vectors $\varphi_h =
\arccos \frac{{| {\mathbf{h}_1^{\dag } \mathbf{h}_2 } |}} {{||
{\mathbf{h}_1 } || || {\mathbf{h}_2 } ||}}$, $\varphi_h \in
[0,\pi/2]$, and the geometric mean $\textsf{SNR}_{geo} = \sqrt
{P_{r1} P_{r2} } || {\mathbf{h}_1 } || || {\mathbf{h}_2 } ||$, the
optimal correlation can be also expressed as
\[
\left| \rho \right|^{opt}  = \min \left( {\frac{{\cos \varphi_h }}
{{\textsf{SNR}_{geo} \sin ^2 \varphi_h }},1} \right).
\]
It is a monotonically decreasing function of $\varphi_h$ and
inversely proportional to $\textsf{SNR}_{geo}$. In Fig.
\ref{fig:rho}, $|\rho|^{opt}$ is drawn for different values of
$\varphi_h$ and $\textsf{SNR}_{geo}$.

For a fixed $\textsf{SNR}_{geo}$ or for a fixed $\varphi_h$, there
exists a threshold above which $|\rho|^{opt}<1$ and below which
$|\rho|^{opt}=1$. The threshold is the solution of $\cos \varphi_h =
\textsf{SNR}_{geo} \sin ^2 \varphi_h$.
\end{remark}

\begin{remark} If $\mathbf{h}_1$ and $\mathbf{h}_2$ are
orthogonal, then $|\rho|^{opt}=0$. If channel vectors are
orthogonal, the differential entropy of the received signal vector
${\bf y}_d$ is maximized when $x_{r1}$ and $x_{r2}$ are
uncorrelated. The resulting sum-rate is given by
\[
 \mathop{\max }\limits_\rho
\log \det (\mathbf{I + HQ}_r \mathbf{H}^\mathbf{\dag } ) = \log (1 +
P_{r1} || {\mathbf{h}_1 } ||^2 )+ \log(1 + P_{r2} || {\mathbf{h}_2 }
||^2 ).
\]
If $\mathbf{h}_1$ and $\mathbf{h}_2$ are parallel, then
$|\rho|^{opt}=1$. If channel vectors are parallel, the differential
entropy of the received signal vector ${\bf y}_d$ is maximized when
$x_{r1}$ and $x_{r2}$ are fully correlated. The resulting sum-rate
is given by
\[
\mathop {\max }\limits_\rho  \log \det (\mathbf{I + HQ}_r
\mathbf{H}^\mathbf{\dag } ) = \log (1 + P_{r1} || {\mathbf{h}_1 }
||^2  + P_{r2} || {\mathbf{h}_2 } ||^2 + 2\sqrt {P_{r1} P_{r2} } |
{\mathbf{h}_1^\mathbf{\dag } \mathbf{h}_2 } |).
\]
\end{remark}

\begin{remark} For a fixed $\bf H$, if either $P_{r1}$ or
$P_{r2}$ are sufficiently small so that $|\rho|^{opt}=1$, the
signaling is optimal when the relay signals are perfectly
correlated. If $P_{r1}, P_{r2} > 0$ and either $P_{r1}$ or $P_{r2}$
are sufficiently large so that $|\rho|^{opt} = \frac{{|\mathbf
{h}_1^\dagger \mathbf {h}_2|}}{{\sqrt {P_{r1} P_{r2}} \det
\left(\mathbf {H^\dagger H}\right)}}<1$, then
$|\rho|^{opt}\rightarrow 0$ as either $P_{r1}$ or $P_{r2}
\rightarrow \infty$. If power is abundant, very small fraction of
relay power needs to be allocated to common signals to satisfy
optimality condition.
\end{remark}

\begin{remark} For $N=1$, the sum rate is maximized by
$|\rho|=1$, i.e., $\alpha=\beta=|\gamma|=1$ and
$\angle\rho=\angle(h_1^* h_2)$ where
$\mathbf{H} = [ {\begin{array}{*{20}c} {h_1 } & {h_2 }  \\
\end{array} } ]$.
\end{remark}

By combining the optimal correlation condition with the achievable
region expression in (\ref{eqn:SIMO_MAC_region}), we get the
following result.

\begin{theorem}[Maximum sum-rate subregion] In the three-dimensional achievable region of a SIMO MAC
given by (\ref{eqn:SIMO_MAC_region}),
the maximum sum-rate subregion is a surface whose boundary
is characterized by
\begin{eqnarray}
\mathcal{R}_{sub} \left( {\left| \rho  \right|^{opt} } \right) =
\bigcup\limits_{\left| \rho  \right|^{opt}  \leq \sqrt\alpha \leq 1}
{\mathcal{R}\left( \alpha  \right)}
\end{eqnarray}
where
\begin{eqnarray}
 \mathcal{R}\left( {\alpha} \right)
  = \left\{ {\left( {R_1 ,R_2 ,R_c } \right)\left| \begin{gathered}
  R_1  \leq \log \left(1 + \overline \alpha  P_{r1} \left\| {\mathbf{h}_1 } \right\|^2 \right) \hfill \\
  R_2  \leq \log \left( {1 + (1-(|\rho|^{opt}/\sqrt\alpha)^2) P_{r2}\left\| {\mathbf{h}_2 } \right\|^2 } \right) \mbox{ for } |\rho|^{opt}>0\hfill \\
  R_2  \leq \log \left( {1 + P_{r2}\left\| {\mathbf{h}_2 } \right\|^2 } \right) \mbox{ for } |\rho|^{opt}=0\hfill \\
  R_1  + R_2  \leq  \log \det \left( {\mathbf{I}  + \mathbf{HQ}_r^p \mathbf{H}^{\dag } } \right) \hfill \\
  R_1  + R_2  + R_c  = \log \det \left( {\mathbf{I}  + \mathbf{HQ}_r \mathbf{H}^{\dag } } \right) \hfill \\
\end{gathered}  \right.} \right\} \label{eqn:sub_region}
\end{eqnarray} where
\[
\mathbf{Q}_r^p  = \left[ {\begin{array}{*{20}c}
   {(1-\alpha) P_{r1} } & 0  \\
   0 & {(1-(|\rho|^{opt}/\sqrt\alpha)^2) P_{r2} }  \\
 \end{array} } \right] \mbox{ for } |\rho|^{opt}>0,
\]
\[
\mathbf{Q}_r^p  = \left[ {\begin{array}{*{20}c}
   {(1-\alpha) P_{r1} } & 0  \\
   0 & { P_{r2} }  \\
 \end{array} } \right] \mbox{ for } |\rho|^{opt}=0,
\]
\[
\mathbf{Q}_r = \left[ {\begin{array}{*{20}c}
   {P_{r1} } & {|\rho|^{opt}e^{j\angle {\mathbf{h}_1^\dagger \mathbf{h}_2}} \sqrt{P_{r1}P_{r2}} }  \\
   {|\rho|^{opt}e^{-j\angle {\mathbf{h}_1^\dagger \mathbf{h}_2}} \sqrt{P_{r1}P_{r2}} } & {P_{r2} }  \\
 \end{array} } \right].
\]
\end{theorem}

$\ $\\
\begin{proof}
To satisfy the optimal correlation, $0 \leq |\rho|^{opt}=
|\gamma|\sqrt{\alpha\beta} \leq 1$, $|\gamma|$, $\sqrt{\alpha}$ and
$\sqrt{\beta}$ should be in the range $[|\rho|^{opt},1]$. In
(\ref{eqn:SIMO_MAC_region}), by setting $|\gamma|=1$, and
$\beta=(|\rho|^{opt}/\sqrt\alpha)^2$ for $|\rho|^{opt}>0$ or
$\beta=0$ for $|\rho|^{opt}=0$, and by taking union over
$|\rho|^{opt} \leq \sqrt\alpha \leq 1$, we characterize the boundary
of the maximum sum-rate surface.
\end{proof}

\begin{example} [Close-to-parallel channel vectors] If
$|\rho|^{opt}=1$, it is required for relays to set
$\alpha=\beta=|\gamma|=1$. Substituting this condition into
(\ref{eqn:SIMO_MAC_region}), we get the expression for the maximum
sum-rate subregion which is a single point given by
\[
\left\{ {\left( {R_1 ,R_2 ,R_c } \right)\left| \begin{gathered}
  R_1  = 0, R_2  = 0, \hfill \\
  R_c  = \log (1 + P_{r1} \left\| {\mathbf{h}_1 } \right\|^2  + P_{r2} \left\| {\mathbf{h}_2 } \right\|^2  + 2\sqrt {P_{r1} P_{r2} } | {\mathbf{h}_1^{\dag } \mathbf{h}_2 } |) \hfill \\
\end{gathered}  \right.} \right\}.
\]
The example achievable rate region is drawn in Fig.
\ref{fig:MAC_region} (a) and (b), and the maximum sum-rate point
$(0,0,R_{sum,MAC}^{max})$ is on the $R_c$ axis.
\end{example}

\begin{example} [Orthogonal channel vectors] If $\left|
\rho  \right|^{opt}  = 0$, at least one of $\alpha$, $\beta$ and
$|\gamma|$ must be zero. We get the expression for the maximum
sum-rate subregion given by
\[
 \left\{ {\left( {R_1 ,R_2 ,R_c } \right)\left| \begin{gathered}
  R_1  \leq \log (1 + P_{r1} \left\| {\mathbf{h}_1 } \right\|^2 ) \hfill \\
  R_2  \leq \log (1 + P_{r2} \left\| {\mathbf{h}_2 } \right\|^2 ) \hfill \\
  R_1  + R_2  + R_c  \leq \sum\limits_{i = 1,2} {\log (1 + P_{ri} \left\| {\mathbf{h}_i } \right\|^2 )}  \hfill \\
\end{gathered}  \right.} \right\}.
\]
As it is shown in Fig. \ref{fig:MAC_region} (d), the maximum
sum-rate surface is the square connecting four points:
$(R_1^{max},R_2^{max},0)$,
$(R_1^{max},0,R_{sum,MAC}^{max}-R_1^{max})$,
$(0,R_2^{max},R_{sum,MAC}^{max}-R_2^{max})$, and
$(0,0,R_{sum,MAC}^{max})$, where $R_i^{max}$ denotes the maximum rate
that can be achieved by each message set. In the subregion, for any
fixed $R_1$ and $R_2$, we can find $R_c$ such that
$R_{sum,MAC}^{max}=R_1+R_2+R_c$.
\end{example}

\subsection{Impact of Common Information Signaling}

Let us discuss how much benefit we can get by having the common
information for $N \geq 2$ case. If the source transmits only the
private signal to relays, i.e. $R_c=0$, the best strategy relays can
do is to have diagonal covariance matrix with individual peak power
$\mathbf{Q}_r^{diag} = diag[ P_{r1},P_{r2}]$. Let
$\mathbf{Q}_r^{opt}$ denote optimal covariance matrix with sum rate
maximizing magnitude and phase angle of $\rho$ derived above. Then,
we get the following result.

\begin{lemma}[Benefit of correlation]\label{th:ortho}
$\log \det \left( {{\bf I} + {\bf HQ}_r^{diag} {\bf H}^\dagger }
\right) \leq \log \det \left( {{\bf I} + {\bf HQ}_r^{opt} {\bf
H}^\dagger } \right)$ with equality if and only if $\mathbf
h_1^\dagger \mathbf h_2 = 0$. The increase in SNR by having
$\mathbf{Q}_r^{opt}$ is given by
\[
\Delta \textsf{SNR} = \left\{ {\begin{array}{*{20}c}
   {\frac{{| \mathbf{
h}_1^{\dag } \mathbf{h}_2 |^2 }} {{\det (\mathbf{H^{\dag }H})}},} & {|\rho|^{opt}<1}  \\
   {2\sqrt {P_{r1} P_{r2} } | {\mathbf{h}_1^\dag
\mathbf{h}_2 } |,} & {|\rho|^{opt}=1}  \\
 \end{array} } \right..
\]
\end{lemma}

\begin{proof}: It is sufficient to show that the sum-rate is a quadratic and concave function of
$|\rho|$, and is monotonically increasing for $0 \leq |\rho| \leq
|\rho|^{opt}$. The function has its minimum at $|\rho|=0$ since
$|\rho|$ is non-negative. When the channel vectors are orthogonal,
the suboptimality vanishes since $|\rho|^{opt}=0$. The SNR increase
can be directly calculated using the result in Lemma
\ref{lemma:optimal_correlation}.
\end{proof}

When the channel column vectors $\mathbf h_1$ and $\mathbf h_2$ are
close to orthogonal, $\mathbf{Q}_r^{diag}$ is almost as good as
$\mathbf{Q}_r^{opt}$. However, when $\mathbf h_1$ and $\mathbf h_2$
are close to parallel, the sum rate by having $\mathbf{Q}_r^{opt}$
at relays shows considerable increase from that by having
$\mathbf{Q}_r^{diag}$. The gain coming from optimal correlation
becomes very large at low SNR. Fig. \ref{fig:Rsum} shows the
examples.

With common information coming from the source, we can introduce
correlation between relays, and they act as if they are in
cooperation. The resulting SIMO MAC behaves like a point-to-point
MIMO channel with per-antenna power constraint.

Here, we can see that there is a minimum required $R_c$ that needs to
be transmitted from the source to relays for achieving maximum
sum-rate in the second hop.

\begin{theorem}[Threshold of $R_c$]
In the SIMO MAC with common information, the threshold of $R_c$
above which a maximum sum-rate point can exist is characterized by
\[
R_c^{th} = \log \frac{{\det \left( {\mathbf{I} + \mathbf{HQ}_r^{opt}
\mathbf{H}^{\dag } } \right)}}
{{\det \left( {\mathbf{I}  + \mathbf{HQ}_r^{k} \mathbf{H}^{\dag } } \right)}} \hfill \\
\] where $\mathbf{Q}_r^{k}  =
diag[(1-k|\rho|^{opt})P_{r1},(1-k^{-1}|\rho|^{opt})P_{r2}]$ and $k =
\sqrt {\frac{{P_{r2} \left\| {\mathbf{h}_2 } \right\|^2  + P_{r1}
P_{r2} \det (\mathbf{H}^{\dag } \mathbf{H})}} {{P_{r1} \left\|
{\mathbf{h}_1 } \right\|^2  + P_{r1} P_{r2} \det (\mathbf{H}^{\dag }
\mathbf{H})}}} $.
\end{theorem}

\begin{proof} In the maximum sum-rate subregion in (\ref{eqn:sub_region}),
after evaluating $\det(\cdot)$ operation, we get the expression for
$R_1+R_2$ given by
\begin{eqnarray}
\begin{gathered}
R_1  + R_2  \leq  \hfill \\
\log \left( {1 + \overline \alpha  P_{r1} \left\| {\mathbf{h}_1 }
\right\|^2  + \left( {1 - \frac{{\left| \rho \right|^2 }} {\alpha }}
\right) P_{r2} \left\| {\mathbf{h}_2 } \right\|^2  + \left( {1 -
\alpha  + \left| \rho  \right|^2  - \frac{{\left| \rho \right|^2 }}
{\alpha }} \right) P_{r1} P_{r2} \det (\mathbf{H}^\mathbf{\dag }
\mathbf{H})} \right). \nonumber
\end{gathered}
\end{eqnarray}
It is straightforward to check that the maximum of $R_1+R_2$ is
given by
\begin{eqnarray}
\begin{gathered}
\mathop {\max }\limits_\alpha (R_1 + R_2) = \log \det (\mathbf{I +
HQ}_r^k \mathbf{H}^{\dag } ) \nonumber
\end{gathered}
\end{eqnarray}
with $\mathbf{Q}_r^{k}$ given in the theorem statement. Finally,
$R_c^{th}=R_{sum,MAC}^{max}-\mathop {\max }\limits_\alpha (R_1 +
R_2) $ results in the minimum possible value of $R_c$ while staying
in the maximum sum-rate subregion.
\end{proof}

\subsection{SISO MAC versus SIMO MAC}
Optimal signaling at DF relays depends on the channel condition. For
a SISO MAC with a single antenna at the destination ($N=1$), the
sum-rate is maximized by $\alpha = \beta = |\gamma|=1$ so that $R_1
= R_2 = 0$ regardless of the channel and power constraints. This is
the case when all the power is allocated to the common signal at
both relays and they add up coherently at the destination. It is
desired for the source to transmit as much common information to
relays as possible, and this strategy maximizes the
source-to-destination sum-rate. Even when $N \geq 2$, if ${\bf h}_1$
and ${\bf h}_2$ are close to parallel in the sense that
$\det(\bf{H^\dagger H})$ is small and $|\rho|^{opt}=1$, fully
correlated relay signals are still optimal.

 In contrast, with multiple antennas at the destination ($N \geq 2$) and $|\rho|^{opt} <1$,
sum-rate maximizing $\alpha$ and $\beta$ depend on both the channel
matrix $\mathbf H$ and relay power constraints. Any combination of
power allocation factors at relays such that $|\gamma|\sqrt{\alpha
 \beta} = \frac{{|\mathbf {h}_1^\dagger \mathbf
{h}_2|}}{{\sqrt {P_{r1} P_{r2}}\det (\bf{H^\dagger H})}}<1$ together
with optimal rotation angle $\angle\gamma=\angle
\mathbf{h}_1^\dagger \mathbf{h}_2$ can maximize the sum rate of a
SIMO MAC. In this case, the source needs to transmit just the right
amount of common information so that signals at two relays are
optimally correlated and the source-to-destination sum-rate is
maximized.

\subsection{The First Hop: MISO BC with Common Information}
In \cite{ScheinThesis}, the first hop is a degraded scalar BC where
one relay with higher SNR can decode both its intend signal and the
signal for the other relay by doing superposition coding and
successive interference cancellation. In this case, correlation
between relay input signals are naturally introduced. In contrast,
our first hop is a non-degraded vector broadcast channel that makes
it possible to send private signals, each decodable by one of the
relays, as well as a common signal decodable by both relays. For
this class of channels, the three dimensional capacity region is
not known, but a good achievable region combining dirty paper coding
(DPC)\cite{costa:83} and superposition was studied in
\cite{JindalGoldsmith:ISIT04,WeingartenSteinbergShamai:ISIT06,WajcerShamaiWiesel:ITW06}
which we also use here.

In this scheme, the transmitting signal ${\bf x}_s$ is a
superposition of three independent signals ${\bf x}_1$, ${\bf x}_2$
and  ${\bf x}_c$, i.e., ${\bf x}_s = {\bf x}_1+{\bf x}_2+{\bf x}_c$,
where ${\bf x}_1, {\bf x}_2, {\bf x}_c \in { \mathbb{C} }^{{\rm N}
\times {\rm 1}}$ denote the signals intended for relay 1, for relay
2 and for both relays, i.e., the common message, respectively. We
assume Gaussian signaling for all signals. Input covariance matrix
is $\mathbf{Q}_s = \mathbf{Q}_1+ \mathbf{Q}_2+ \mathbf{Q}_c$, where
$\mathbf{Q}_j = \mathbb{E}[ {\mathbf{x}_j \mathbf{x}_j^\dagger } ]$,
$j \in \{1,2,c\}$.

Common information is decoded at both relays before decoding private
messages. Private messages are encoded using dirty paper coding,
i.e., the private message for relay 1 is first encoded as
$\mathbf{x}_1$, and the private message for relay 2 is then encoded
as $\mathbf{x}_2$ using $\mathbf{x}_1$ as side information so that
$\mathbf{x}_2$ can be decoded at relay 2 without interference from
$\mathbf{x}_1$. The encoding order can be reversed. With this
scheme, an achievable rate region is given by
\begin{eqnarray}
\mathcal{R}_{BC} \left( P_s \right) = Co\left( {\bigcup\limits_{\pi
,\mathbf{Q}_s : \text{tr}\left( {\mathbf{Q}_s } \right) \leq P_s}
{\mathcal{R}\left( {\pi ,\mathbf{Q}_s } \right)} } \right)
\end{eqnarray}
where ${\mathcal{R}\left( {\pi ,\mathbf{Q}_s } \right)} $ is the
achievable region for a given encoding order $\pi  \in \left\{ {\pi
_{12} ,\pi _{21} } \right\}$ and input covariance matrix
$\mathbf{Q}_s$, where $Co(\cdot)$ is the convex hull operator. If
$\mathbf {x}_2$ is encoded first, for example, then we have
\begin{eqnarray}
\mathcal{R}\left( {\pi _{12} ,\mathbf{Q}_s } \right) = \left\{
{\left( {R_1 ,R_2 ,R_c } \right)\left|
\setlength{\arraycolsep}{.1em} {\begin{array}{*{20}c}
   {R_1  \leq  \log\left(1 + {\mathbf{g}_1 \mathbf{Q}_1 \mathbf{g}_1^\dagger  } \right),} \hfill \\
   {R_2  \leq \log\left(1 + {\frac{{\mathbf{g}_2 \mathbf{Q}_2 \mathbf{g}_2^\dagger  }}
{{1 + \mathbf{g}_2 \mathbf{Q}_1 \mathbf{g}_2^\dagger  }}} \right),} \hfill \\
   {R_c  \leq  \mathop {\min }\limits_{i \in \left\{ {1,2} \right\}} \log\left(1 + {\frac{{\mathbf{g}_i \mathbf{Q}_c \mathbf{g}_i^\dagger  }}
{{1+\mathbf{g}_i \left( {\mathbf{Q}_1  + \mathbf{Q}_2 } \right)\mathbf{g}_i^\dagger  }}} \right)} \hfill \\
 \end{array} } \right.} \right\}.
\end{eqnarray}

Fig. \ref{fig:BC_region}. (a) depicts an example of an achievable
region where two row vectors of $\bf G$ are parallel and linearly
dependent so that $\bf G$ is ill-conditioned and rank-deficient.
Note that in the figures, $\varphi_g = \arccos \frac{{|
{\mathbf{g}_1 \mathbf{g}_2^{\dag } } |}} {{\left\| {\mathbf{g}_1 }
\right\|\left\| {\mathbf{g}_2 } \right\|}}$. One can see that the
maximum sum-rate surface is a plane connecting three points:
$(R_{1,BC}^{max},0,0)$, $(0,R_{2,BC}^{max},0)$, and
$(0,0,R_{c,BC}^{max})$ where $R_{j,BC}^{max}$ denotes the maximum
rate achieved by allocating all power to $\mathbf {x}_j$ so that
$\mathbf {Q}_j = \mathbf{Q}_s$:
\[
R_{i,BC}^{max} = \mathop {\max }\limits_{\mathbf{Q}_s} \log\left(1 +
{\mathbf{g}_i \mathbf{Q}_s \mathbf{g}_i^\dagger  } \right) = \log
\left( {1 + P_s \left\| {\mathbf{g}_i } \right\|^2}\right)
\]
\[
R_{c,BC}^{max} = {\mathop {\max }\limits_{\mathbf{Q}_s} \mathop
{\min }\limits_{i} \log(1+ {\mathbf{g}_i \mathbf{Q}_s
\mathbf{g}_i^\dagger })} = \mathop {\min }\limits_{i} \log \left( {1
+ P_s \left\| {\mathbf{g}_i } \right\|^2}\right)
\]
where $i \in \{1,2\}$. In this specific example, the channel is
symmetric in the sense that $\left\|\mathbf
{g}_1\right\|=\left\|\mathbf {g}_2\right\|$. In fact, this is the
only case where having common information does not incur sum-rate
penalty.

As the opposite extreme, Fig. \ref{fig:BC_region}. (b) depicts an
example of an achievable region where two row vectors of $\bf G$ are
orthogonal and linearly independent so that $\bf G$ is
well-conditioned and full-rank. In the symmetric example, the point
that achieves the maximum sum-rate
\[
R_{sum,BC}^{max}= \log (1+\frac{{P_s}}{{2}} \| \mathbf g_1 \|) +
\log (1+\frac{{P_s}}{{2}} \| \mathbf g_2 \|)
\]
is on the line $R_1=R_2$ and $R_c=0$. The point that has the minimum
sum-rate on the boundary is on the $R_c$ axis, i.e., $R_1=R_2=0$ and
\[
R_{c,BC}^{max} = {\mathop {\max }\limits_{\mathbf{Q}_s} \mathop
{\min }\limits_{i} \log(1+ {\mathbf{g}_i \mathbf{Q}_s
\mathbf{g}_i^\dagger })} = \log (1+\frac{{P_s}}{{2}} \| \mathbf g_1
\|) = \log (1+\frac{{P_s}}{{2}} \| \mathbf g_2 \|)
\]
where we can see the sum-rate penalty due to beamforming
inefficiency.

Note that the maximum sum-rate points of a MISO BC are always on the
$R_1-R_2$ plane. In fact, they correspond to the dominant face of the
two user BC achievable region by DPC without common information.

\subsection{Achievable Rate by DF}

For the GPRN drawn in Fig. \ref{f:systemmodel}, a triplet
$\left(R_1,R_2,R_c\right)$ is said to be achievable by DF if it
belongs to the intersection of the rate regions of the first hop
MISO BC and the second hop SIMO MAC. In this context, the maximum
rate by DF can be defined by
\begin{eqnarray}
R_{DF}^{max} =  \mathop {\max }\limits_{\left( {R_1 ,R_2 ,R_c }
\right) \in \mathcal{R}_{DF}} R_1+R_2+R_c \label{eqn:SumRate}
\end{eqnarray}
where $\mathcal{R}_{DF}=\mathcal{R}_{BC} \left( P_s \right) \cap
\mathcal{R}_{MAC} \left( {P_{r1} ,P_{r2} } \right)$. Fig.
\ref{fig:BCMAC}. shows examples of the rate regions of MISO BC and
SIMO MAC, the intersection of which is the achievable rate region by
DF.

If the source-to-relay link SNR is high enough, the second hop
becomes the bottleneck and determines the source-to-destination
sum-rate.

\begin{theorem}[Optimality condition of DF]
If there is a rate triple $\left(R_1,R_2,R_c\right) \in
\mathcal{R}_{sub} \left( {\left| \rho  \right|^{opt} } \right)$
which is included in the MISO BC region $\mathcal{R}_{BC} \left( P_s
\right)$, then $R_{DF}^{max} = R_{sum,MAC}^{max}$ meets the upper
bound and determines the capacity of the vector Gaussian parallel
relay network.
\end{theorem}

\begin{proof}
Let us first consider $|\rho|^{opt}=1$ case where there is a single
maximum sum-rate point in the SIMO MAC region. At $|\rho|=1$, the
term $\log \det \left( {\mathbf{I} + \mathbf {HQ}_r
\mathbf{H}^\dagger } \right)$ in
(\ref{eqn:correlation_bound1}) and (\ref{eqn:correlation_bound2})
is maximized
and is smaller than (\ref{eqn:broadcast_cut-set}) since the
following relationships are hold:
\[
R_{sum,MAC}^{max} \leq R_{c,BC}^{max} \leq R_{i,BC}^{max} (i \in
\{1,2\}) \leq R_{sum,BC}^{max} \leq R_{BC}^{upper}.
\]
We can achieve the tightest upper bound $R_{sum,MAC}^{max}$ by
allocating all relay power to the common signal:
$\alpha=\beta=|\gamma|=1$.

Now we consider $|\rho|^{opt}<1$ case where there exist more than
one maximum sum-rate point in the SIMO MAC region. It is
tedious but easy to verify that, at $|\rho|=|\rho|^{opt}$,
\[
R_{sum,MAC}^{max} \leq R_{sum,BC}^{max} \leq R_{BC}^{upper}
\]
\[
2\log \det \left( {\mathbf{I} + \mathbf {HQ}_r \mathbf{H}^\dagger }
\right) \leq \sum\limits_{i = 1,2} \log( {1 + P_s \left\|
{\mathbf{g}_i } \right\|^2 })+\sum\limits_{i = 1,2} \log\left( 1 +
(1-|\rho|^2) P_{ri} \left\| {\mathbf{h}_i } \right\|^2  \right).
\]
We can achieve the tightest upper bound $R_{sum,MAC}^{max}$ by
optimal power allocation at relays such that $|\rho|^{opt}=|\gamma|
\sqrt{\alpha \beta}$.
\end{proof}

With DF, the first and second hops are completely separated in the
sense that after finishing the first stage of transmission, relays
start a new stage of transmission by encoding the received
information again. Here, the right approach is first to figure out
what is optimal in the second hop, and then to check if the optimal
operating point, i.e., one of the maximum sum-rate points of the
SIMO MAC is supportable by the first hop. If the optimal point is
achievable, the network nodes would start communication by setting
parameters to satisfy optimality conditions. If none of the maximum
sum-rate points of the SIMO MAC is achievable, then the nodes would
try to find the operating point as close to the optimal as possible.

What if the relay-to-destination link SNRs are high enough so that the
first hop is the bottleneck? If one of the maximum sum-rate points
of the MISO BC on the $R_1-R_2$ plane is included in
$\mathcal{R}_{MAC}$, then $R_{DF}^{max} = R_{sum,BC}^{max}$. In this
case, the broadcast cut-set upper bound in
(\ref{eqn:broadcast_cut-set}) is tighter than the others. It turns out
that the upper bound and the achievable rate meet in some special
cases where the first hop row vectors are orthogonal as will be seen
in numerical results in Section VII. However, they do not meet in
general, which implies suboptimality of DF in the case. If the first
hop is the bottleneck, full decoding at relays gives too much
restriction, and AF and CF schemes can do better.

\subsection{Symmetric Channels}
In this subsection, we narrow down our attention to the symmetric
channels: $\left\| {\mathbf{g}_1 } \right\| = \left\| {\mathbf{g}_2
} \right\|$, $\left\| {\mathbf{h}_1 } \right\| = \left\|
{\mathbf{h}_2 } \right\|$, $P_{r1} = P_{r2}$, and $R_1 = R_2$. We
consider four examples illustrated in Fig. \ref{fig:Symmetric}.
where $R_p = R_1 + R_2$.

In Fig. \ref{fig:Symmetric}. (a), the straight line $\overline{AB}$
and the curved line $\overline{BC}$ constitute the surface of the MAC
achievable rate region. The curved line $\overline{EG}$ is the
surface of the BC achievable rate region. Here, we have the source power
constraint $P_s$ so large that the upper bound
(\ref{eqn:broadcast_cut-set}) and the second terms of $min$ in
(\ref{eqn:correlation_bound1}) and (\ref{eqn:correlation_bound2})
are loose. In this case, $\log \det \left( {\mathbf{I}  +
\mathbf{HQ}_r \mathbf{H}^\dag  } \right)$ is the active upper bound,
and the sum-rate constraint $R_1+R_2+R_c$ is the straight line that
goes through the points $A$, $B$, and $D$. The straight line
$\overline{AB}$ is the maximum sum-rate subregion of the MAC, and $F$ is
the crossing point of the BC and MAC surfaces. All the points on the
line $\overline{FB}$ achieve the capacity of the vector Gaussian
parallel relay network as they are in the BC and MAC achievable rate
regions and meet the sum-rate upper bound.

Other things being equal, the BC rate region with a smaller power
constraint is drawn in Fig. \ref{fig:Symmetric}. (b). In this
example, the crossing point $F$ coincides with the point $B$. Thus,
there exists a single capacity achieving point at $B=F$. In Fig.
\ref{fig:Symmetric}. (c), the BC rate region gets even smaller, and
the DF maximum sum-rate point $F$ does not meet the sum-rate upper
bound. Thus, in this case, DF does not achieve the capacity of the
vector Gaussian parallel relay network. Finally, Fig.
\ref{fig:Symmetric}. (d) illustrates the case where the source power
constraint $P_s$ is so small that $\log \det \left( {\mathbf{I} +
\mathbf{HQ}_r \mathbf{H}^\dag  } \right)$ in the upper bound
expression is not active anymore.

\section{Amplify-and-Forward (AF)}
We have seen that if the relay-to-destination link SNR is high enough, the first
hop is the bottleneck and DF does not satisfy the
optimality condition that
requires
at least one maximum sum-rate point of
SIMO MAC should be inside MISO BC achievable region. Instead of
requiring signals with low SNR to be
decoded at relays, it would be better for relays to just forward their
received signals to the destination so that the benefit of high SNR
in the second hop is maximally utilized.

\subsection{Achievable Rate of AF}
The received signal at relays can
be expressed in vector notation by
\[
\mathbf{y}_r = \mathbf{Gx}_s  + \mathbf{n}_r.
\]
With AF, the relays just amplify their received signals before
forwarding them to the destination, and the transmit signal vector of the relays is
given by
\[
\mathbf{x}_r  = \mathbf{Ay}_r
\]
where $\mathbf{A} = diag[ae^{j\phi},b]$. The amplification factors
should be in the range
\[
0 \leq a \leq a^{peak}=\sqrt {\frac{{P_{r1} }} {{1 + \mathbf{g}_1
\mathbf{Q}_s \mathbf{g}_1^{\dag } }}} ,\ \ 0 \leq b \leq
b^{peak}=\sqrt {\frac{{P_{r2} }} {{1 + \mathbf{g}_2 \mathbf{Q}_s
\mathbf{g}_2^{\dag } }}}
\]
because of the power constraints at the relays. The received signal at the
destination is given by
\[
\begin{gathered}
  \mathbf{y}_d = \mathbf{Hx}_r  + \mathbf{n}_d  = \mathbf{HAy}_r  + \mathbf{n}_d  \hfill \\
   = \mathbf{HAGx}_s  + \mathbf{HAn}_r  + \mathbf{n}_d  \hfill \\
   = \mathbf{HAGx}_s  + \mathbf{n}_e  \hfill \\
   = (ae^{j\phi } \mathbf{h}_1 \mathbf{g}_1  + b
\mathbf{h}_2 \mathbf{g}_2 )\mathbf{x}_s + (ae^{j\phi } n_{r1}
\mathbf{h}_1  + b n_{r2} \mathbf{h}_2 + \mathbf{n}_d )
\end{gathered}
\]
where ${\bf n}_e$ denotes the total effective noise. Since the noises
added at relay receivers are also amplified and forwarded through
the channel, the noise vector ${\bf n}_e$ is spatially non-white in
general. Noise covariance matrix is symmetric and can be decomposed
as
\[
\mathbf{K} = \mathbb{E} [\mathbf{n}_e \mathbf{n}_e^{\dag }] =
\mathbf{I + HAA^\dag H^\dag} = \mathbf{U^\dag \Lambda U}
\]
where $\Lambda$ is a diagonal matrix with eigenvalues of $\bf K$ as
its diagonal elements, and $\bf U$ is a unitary matrix. Then, the channel can be
transformed into an equivalent white noise channel given by
\[
  \mathbf{\Lambda }^{ - 1/2} \mathbf{Uy}_d  = \mathbf{\Lambda }^{ - 1/2} \mathbf{UHAGx}_s  + \mathbf{\Lambda }^{ - 1/2} \mathbf{Un}_e  \hfill \\
   = \mathbf{Fx}_s  + \mathbf{n}_w
\]
where $\bf F$ and ${\bf n}_w$ denote the effective source-to-destination
channel matrix and the white noise vector, respectively. For a fixed
$\bf A$, this is a point-to-point MIMO channel whose maximum rank
is $2$ limited by the number of relays.
If all the input signal is Gaussian, we get the expression for an
achievable rate by AF given by
\begin{eqnarray}
\begin{gathered}
R_{AF}  \leq I(\mathbf{x}_s ;\mathbf{y}_d ) \hfill \\
=\log \frac{{\det \left( {\mathbf{I} + \mathbf{HAA}^{\dag }
\mathbf{H}^{\dag }  + \mathbf{HAGQ}_s \mathbf{G}^{\dag }
\mathbf{A}^{\dag } \mathbf{H}^{\dag } } \right)}}
{{\det \left( {\mathbf{I} + \mathbf{HAA}^{\dag } \mathbf{H}^{\dag } } \right)}} \hfill \\
   = \log \det \left( {\mathbf{I} + \mathbf{\Lambda }^{ - 1/2} \mathbf{UHAGQ}_s \mathbf{G}^{\dag } \mathbf{A}^{\dag } \mathbf{H}^{\dag } \mathbf{U}^{\dag } \mathbf{\Lambda }^{ - 1/2} } \right) \hfill \\
   = \log \det \left( {\mathbf{I} + \mathbf{FQ}_s \mathbf{F}^{\dag } }
   \right). \hfill \label{eqn:AF_rate}
\end{gathered}
\end{eqnarray}
In order to get the maximum achievable rate, the source signal
covariance matrix $\mathbf{Q}_s$ and the relay amplification matrix
$\bf A$ need to be jointly optimized.

\begin{theorem}[Asymptotic Optimality of AF]
If $\bf H$ is full rank, AF is asymptotically optimal in the high relay power limit
in the sense that
\[
\mathop {\lim }\limits_{P_{r1}  = P_{r2}  \to \infty } (\mathop
{\max }\limits_{\mathbf{Q}_s } \log \det \left( {\mathbf{I}  +
\mathbf{GQ}_s \mathbf{G}^{\dag } } \right) -\mathop {\max
}\limits_{\mathbf{Q}_s, \mathbf{A} }R_{AF} )= 0.
\]
\end{theorem}

$\\ $
\begin{proof} By rearranging (\ref{eqn:AF_rate}), we get
\[
\begin{gathered}
  R_{AF}  = \log \det \left( \mathbf{I + HAGQ}_s \mathbf{G^\dag  A^\dag  H^\dag  \left( {I + HAA^\dag  H^\dag  } \right)^{ - 1} } \right) \hfill \\
   = \log \det \left( \mathbf{I + GQ}_s \mathbf{G^\dag  A^\dag  H^\dag  \left( {I + HAA^\dag  H^\dag  } \right)^{ - 1} HA} \right) \hfill \\
   = \log \det \left( \mathbf{I + GQ}_s \mathbf{G^\dag  \left( {I - \left( {I + A^\dag  H^\dag  HA} \right)^{ - 1} } \right)} \right) \hfill \\
\end{gathered}
\]
where we use matrix inversion lemma. We first set $a$ and $b$ to
peak values under relay power constraints, assume $\phi$ is
optimally chosen, and let $R_{AF}^{peak}$ denote the corresponding
rate. By showing that if $\det (\mathbf{H}^\mathbf{\dag }
\mathbf{H})>0$, $\left( \mathbf{I+A^\dag H^\dag HA} \right)^{ - 1}
\to \mathbf{0}$ as $a \to \infty$ and $b \to \infty$ where
\[
\left( \mathbf{I+A^\dag H^\dag HA} \right)^{ - 1}  = \frac{{\left[
{\begin{array}{*{20}c}
   {1 + b^2 \left\| {\mathbf{h}_2 } \right\|^2 } & { - abe^{-j\phi}\mathbf{h}_1^\mathbf{\dag } \mathbf{h}_2 }  \\
   { - abe^{j\phi}\mathbf{h}_2^\mathbf{\dag } \mathbf{h}_1 } & {1 + a^2 \left\| {\mathbf{h}_1 } \right\|^2 }  \\
 \end{array} } \right]}}
{{1 + a^2 \left\| {\mathbf{h}_1 } \right\|^2  + b^2 \left\|
{\mathbf{h}_2 } \right\|^2  + a^2 b^2 \det (\mathbf{H}^\mathbf{\dag
} \mathbf{H})}},
\]
we get the following result
\[
\mathop {\lim }\limits_{P_{r1}  = P_{r2}  \to \infty } \left( \log
\det \left( {\mathbf{I}  + \mathbf{GQ}_s \mathbf{G}^\mathbf{\dag } }
\right) -R_{AF}^{peak} \right)= 0
\]
Then, the theorem statement naturally follows since by definition,
for a fixed ${\bf Q}_s$,
\[
R_{AF}^{peak}  \leq \mathop {\max }\limits_{\mathbf{A} }R_{AF}.
\]

\end{proof}

\subsection{Iterative Optimization Algorithm for AF} We can maximize
$R_{AF}$ by an iterative algorithm as follows. For optimizing
$\mathbf{Q}_s$, we apply singular value decomposition (SVD) to $\bf
F$ and waterfilling over two parallel scalar channels with non-zero
singular values \cite{Telatar:1999}. First, we define the covariance
matrix of ${\bf Gx}_s$ and the rate function of $(a,b)$ given by
\[
\mathbf{Q}^e = \mathbb{E}[{\bf G x}_s {\bf x}_s^\dagger {\bf
G}^\dagger] = \mathbf{GQ_sG^\dag } = \left[ {\begin{array}{*{20}c}
   {q_{11} } & {q_{12} }  \\
   {q_{21} } & {q_{22} }  \\
 \end{array} } \right],
\]
\[
\begin{gathered}
R(a,b) = \hfill \\
\log \left( {1 + \frac{{a^2 b^2 \det (\mathbf{H}^{\dag }
\mathbf{H})(\text{tr}(\mathbf{Q}^e ) + \det (\mathbf{Q}^e )) + a^2
\left\| {\mathbf{h}_1 } \right\|^2 q_{11}  + b^2 \left\|
{\mathbf{h}_2 } \right\|^2 q_{22}  + 2ae^{-j\phi}b
{\mathbf{h}_1^{\dag } \mathbf{h}_2 q_{21} } }} {{1 + a^2 \left\|
{\mathbf{h}_1 } \right\|^2  + b^2 \left\| {\mathbf{h}_2 } \right\|^2
+ a^2 b^2 \det (\mathbf{H}^{\dag } \mathbf{H})}}} \right).
\end{gathered}
\]

$\ $\\
Next, we do the following steps:

\textbf{Step 0)} Set $\mathbf{Q}_s=diag[P_s/2,P_s/2]$.\\

\textbf{Step 1)} Calculate $\mathbf{Q}^e =\mathbf{GQ}_s \mathbf{G}^\dag$.\\
${}\ \ \ \ \ \ \ \ \ \ \ \ \ \ \ \ $Set
$a^{peak}=\sqrt{\frac{{P_{r1}}}{{1+q_{11}}}}$,
$b^{peak}=\sqrt{\frac{{P_{r2}}}{{1+q_{22}}}}$, and $\phi = \angle
(\mathbf{h}_1^{\dag } \mathbf{h}_2 q_{21}
)$.\\

\textbf{Step 2)} Calculate $\mathop {\max }\limits_a R(a,b^{peak})$
subject to $0 \leq a \leq a^{peak}$ by solving $\frac{\partial } {{\partial a}}R(a,b^{peak}) = 0$.\\
${}\ \ \ \ \ \ \ \ \ \ \ \ \ \ \ \ $Calculate $\mathop {\max
}\limits_b R(a^{peak},b)$ subject to $0 \leq b \leq b^{peak}$
by solving $\frac{\partial } {{\partial b}}R(a^{peak},b) = 0$.\\
${}\ \ \ \ \ \ \ \ \ $Between $(a^{opt},b^{peak})$ and
$(a^{peak},b^{opt})$, choose the one that results in a higher
rate.\\

\textbf{Step 3)} Set ${\bf A} = diag[ae^{j\phi},b]$ using the values
obtained above.\\
${}\ \ \ \ \ \ \ \ \ \ \ \ \ \ \ \ $Calculate $\mathbf{K}$ and do
eigenvalue decomposition to get $\bf U$ and $\Lambda$ such that
\[
\mathbf{K = I + HAA^\dag H^\dag = U^\dag \Lambda U}.
\]
${}\ \ \ \ \ \ \ \ \ \ \ \ \ \ \ \ $Calculate ${\bf F} =
\mathbf{\Lambda }^{ -
1/2} \mathbf{UHAG}$.\\

\textbf{Step 4)} Optimize $\mathbf{Q}_s$ via singular value
decomposition of $\bf F$ and waterfilling \cite{Telatar:1999}.\\
${}\ \ \ \ \ \ \ \ \ \ \ \ \ \ \ \ $Update $\mathbf{Q}_s$ and
calculate
$R_{AF}= \log \det \left( {\mathbf{I + FQ}_s \mathbf{F^\dag} }\right)$.\\

\textbf{Step 5)} Terminate if $R_{AF}$ already converged to a
certain value.\\
${}\ \ \ \ \ \ \ \ \ \ \ \ \ \ \ \ $Otherwise, go to Step 1.

Before closing the section, it is worth noting that full power
transmission sometimes hurts. This is the case when one of the
relays has received a signal with very low SNR so that transmission
at full power degrades the received SNR at the destination. Fig.
\ref{fig:rate_b} shows examples. For each of the two different
values of $\left\|{\bf g}_2\right\|$ shown in the figure, we run the
algorithm steps from 0 to 4 just once, and draw the resulting
$R(a^{peak},b)$ versus $b$ curves. For the curve in Fiq.
\ref{fig:rate_b}. (a), the maximum of the curve indicated by a
circle happens at a point of $b$ above $b^{peak}$ indicated by
vertical line. It means that the received SNR at the second relay is
still high, and full power transmission helps. However, for the
curve in Fiq. \ref{fig:rate_b}. (b), as the maximum happens at a
point less than $b^{peak}$, the received SNR at the second relay is
too low for transmission at full power to be optimal.

\section{Compress-and-Forward (CF)}
For CF, relays compress or quantize their received signals,
re-encode, and forward them to the destination. At the destination,
the decoder tries to decode the received signal to recover relay
input signals, and finally decompress the relay signals to recover
the information transmitted from the source. CF achievable rates
were first derived by applying Wyner-Ziv source coding
\cite{WynerZiv} to a classical one relay model in
\cite{CoverElGamal:79}, their extension to multiple relay models in
\cite{KramerGastparGupta:2005}. The derivation for a special case of
the Gaussian parallel relay network with $N=M=1$ can be found in
\cite{ScheinThesis} and \cite{XueSandhu}. The extension to our
network model in terms of mutual information is straightforward as
follows
\begin{eqnarray}
\begin{gathered}
R_{CF}  \leq I(\mathbf{x}_s ;\hat{y}_{r1}, \hat{y}_{r2}) \hfill \\
I(\hat{y}_{r1} ; y_{r1}| \hat{y}_{r2})  \leq I(x_{r1} ;\mathbf{y}_d| x_{r2}) \hfill \\
I(\hat{y}_{r2} ; y_{r2}| \hat{y}_{r1})  \leq I(x_{r2} ;\mathbf{y}_d| x_{r1}) \hfill \\
I(\hat{y}_{r1}, \hat{y}_{r2} ; y_{r1}, y_{r2})  \leq I(x_{r1}, x_{r2} ;\mathbf{y}_d). \hfill \\
\label{eqn:CF_rate}
\end{gathered}
\end{eqnarray}

By introducing quantization noise, we have a compressed version of
relay signals given by
\begin{eqnarray}
\begin{gathered}
\hat y_{r1}=y_{r1} + \sqrt a \hat n_{r1} \hfill \\
\hat y_{r2}=y_{r2} + \sqrt b \hat n_{r2} \hfill \\
\end{gathered}
\end{eqnarray} where $a, b > 0$. Assuming the quantization noise and
all input signals are Gaussian distributed, we can evaluate the
mutual information expressions to have the following result.

\begin{theorem} With CF, the following rate is achievable in the vector Gaussian parallel relay network.
\begin{eqnarray}
\begin{gathered}
R_{CF} = \mathop {\max }\limits_{\mathbf{Q}_s, \mathbf{A} } \log
\frac{{\det \left( {\mathbf{I} + \mathbf{A} + \mathbf{GQ}_s
\mathbf{G}^{\dag }} \right) }} {{ (1+a)(1+b)}}, \hfill \\
\mbox{subject to } \log \frac{{\det \left( {\mathbf{I} + \mathbf{A}
+ \mathbf{GQ}_s \mathbf{G}^{\dag }} \right) }} {{ a(
1+b+\mathbf{g}_2 \mathbf{Q}_s \mathbf{g}_2^\dag )}} \leq \log (1 +
P_{r1} \left\| {\mathbf{h}_1 } \right\|^2 ),\hfill \\
\mbox{\ \ \ \ \ \ \ \ \ \ \ \ }\log \frac{{\det \left( {\mathbf{I} +
\mathbf{A} + \mathbf{GQ}_s \mathbf{G}^{\dag }} \right) }} {{ b(
1+a+\mathbf{g}_1 \mathbf{Q}_s \mathbf{g}_1^\dag )}} \leq \log (1 +
P_{r2} \left\| {\mathbf{h}_2 } \right\|^2 ),\hfill \\
\mbox{\ \ \ \ \ \ \ \ \ \ \ \ }\log \frac{{\det \left( {\mathbf{I} +
\mathbf{A} + \mathbf{GQ}_s \mathbf{G}^{\dag }} \right) }} {{ ab }}
\leq \log \det \left( {\mathbf{I} + \mathbf{HQ}_r \mathbf{H}^{\dag
}} \right),\hfill \\
\mbox{\ \ \ \ \ \ \ \ \ \ \ \ \ \ \ \ \ \ \ \ \ \ \ \ \ \ \ \ \ \ \ \ \ \ \ \ } {\rm tr}\left( {{\bf Q}_s } \right) \leq P_s, \hfill \\
\mbox{\ \ \ \ \ \ \ \ \ \ \ \ \ \ \ \ \ \ \ \ \ \ \ \ \ \ \ \ \ \ \ \ \ \ \ \ \ \ \ \ } a,b \geq 0 \hfill \\
\end{gathered}
\end{eqnarray} where $\mathbf{Q}_s = \mathbb{E}[\mathbf{x}_s \mathbf{x}_s^\dag]$, $\mathbf{Q}_r = \mathbb{E}[\mathbf{x}_r \mathbf{x}_r^\dag] =
diag[P_{r1},P_{r2}]$, and $\mathbf{A} = diag[a,b]$.

\end{theorem}

\begin{proof}
It is straightforward to show the theorem result by evaluating the
mutual information expressions with the assumption that the input
distributions are circularly symmetric complex Gaussian.
\end{proof}

Similar to the AF scheme, the CF achievable rate becomes close to the
upper bound as the relay power goes to infinity. We get the
following result for the CF achievable rate.
\begin{theorem}[Asymptotic Optimality of CF]
In the vector Gaussian parallel relay network, regardless of the
rank of $\bf H$, CF is asymptotically optimal in the high relay
power limit in the sense that
\[
\mathop {\lim }\limits_{P_{r1}  = P_{r2}  \to \infty } (\mathop
{\max }\limits_{\mathbf{Q}_s } \log \det \left( {\mathbf{I}  +
\mathbf{GQ}_s \mathbf{G}^{\dag } } \right) - R_{CF} )= 0.
\]
\end{theorem}

$\\ $
\begin{proof}
As $P_{r1}$ and $P_{r2}$ go to infinity, the optimization of
$R_{CF}$ becomes unconstrained. The objective is maximum when
$a=b=0$. Thus,
\begin{eqnarray}
\begin{gathered}
\mathop {\max }\limits_{\mathbf{Q}_s, \mathbf{A} } \log \det \left(
{\mathbf{I} + \mathbf{GQ}_s \mathbf{G}^{\dag } (\mathbf{I} +
\mathbf{A})^{-1}} \right) \to \mathop {\max }\limits_{\mathbf{Q}_s}
\log \det \left( {\mathbf{I} + \mathbf{GQ}_s \mathbf{G}^{\dag }}
\right).
\end{gathered}
\end{eqnarray}
\end{proof}

\section{Numerical Results}

In this section, we consider a few numerical examples to compare
achievable rates by different schemes and the upper bound derived
throughout the paper. Let us pick three symmetric matrices for
$\mathbf G$ (or $\mathbf H$):

\[
\left[ {\begin{array}{*{20}c}
   {1} & {0}  \\
   {0} & {1}  \\
\end{array} } \right],
\left[ {\begin{array}{*{20}c}
   {0.9285} & {0.3714}  \\
   {0.3714} & {0.9285}  \\
\end{array} } \right], \mbox{ and }
\left[ {\begin{array}{*{20}c}
   {0.7071} & {0.7071}  \\
   {0.7071} & {0.7071}  \\
\end{array} } \right].
\]
The angles between channel row (or column) vectors are $90^\circ$,
$46.3972^\circ$, and $0^\circ$, respectively. Fig.~\ref{f:sum-rate}
shows the results of six different combinations of the first and
second hop channel matrices where the achievable rates by DF, AF and
CF, and the upper bound are plotted.

As we have investigated in Section IV, below a certain level of the
relay-to-destination link SNR, DF achievable rate meets the upper
bound. The threshold point at which DF starts achieving the capacity
can be calculated from the DF optimality condition. In Fig.
\ref{f:sum-rate}. (a) and (b), we can see that when the first hop
channel vectors are orthogonal, DF always performs better than AF
and CF, and achieves the capacity in the high relay power regime. In
contrast, when the channel vectors are not orthogonal as in Fig.
\ref{f:sum-rate}. (c), (d), (e) and (f), DF achievable rates are
bounded away from the upper bound in the high relay power regime.

When the second hop channel matrix is full rank as in Fig.
\ref{f:sum-rate}. (a), (c) and (e), AF is shown to asymptotically
achieve the capacity in the high relay power limit. In Fig.
\ref{f:sum-rate}. (b) and (d), AF achievable rate stays away from
the capacity even in the high relay power limit since the second hop
channel is rank-deficient. In Fig. \ref{f:sum-rate}. (f), again, AF
becomes asymptotically optimal even though the second hop channel
vectors are not full rank. In the case, as the first hop channel is
already rank-deficient, there is no additional penalty by the
rank-deficient second hop. CF seems to be advantageous over AF in
the sense that it asymptotically achieve the capacity in the high
relay power limit regardless of the rank of the second hop channel.

\section{Conclusion}
Throughout the paper, we have shown how much rate is achievable by
DF, AF or CF, and when the achievable rates meet the upper bound.
The relative advantage of each scheme varies depending not only on
which hop is the bottleneck but also on the ranks of the first and
second hop channel matrices. The effect of the channel rank is newly
explained in our work.

For the DF relaying, we used a combination of a MISO broadcast
scheme and a SIMO multiple access scheme, with which a few
interesting characteristics of the SIMO MAC are investigated. It is
shown that DF achieves the capacity in the low relay power regime.

Earlier results for AF and CF were extended to explain our vector
Gaussian network and to compare their achievable rates to that of
DF. AF was shown to achieve close-to-capacity rate in the high relay
power regime when the second hop channel matrix is full rank while
CF similarly achieves the asymptotic capacity regardless of the
channel rank.

\appendices
\section{Proof of Lemma 1}
First, we shall find the optimal angle of $\rho$.
By differentiating
$\det \left( {\mathbf{I} + \mathbf{HQ}_r \mathbf{H}^{\dag } } \right)
= \det \left( {\mathbf{I} + \mathbf{Q}_r \mathbf{H}^{\dag }
\mathbf{H}} \right)$
with respect to $\theta$ and setting it to
zero, we get
\begin{eqnarray}
e^{j\theta } \mathbf{h}_2 ^\mathbf{\dag } \mathbf{h}_1  = e^{ -
j\theta } \mathbf{h}_1 ^\mathbf{\dag } \mathbf{h}_2. \nonumber
\end{eqnarray}
Since the left hand side is the conjugate of the right hand side, they
both should be real, which means the optimal $\theta$ needs to satisfy
\begin{eqnarray}
\theta ^{opt} =  - \angle (\mathbf{h}_2 ^\mathbf{\dag }
\mathbf{h}_1)  = \angle (\mathbf{h}_1 ^\mathbf{\dag } \mathbf{h}_2).
\nonumber
\end{eqnarray}
Using this optial angle, we can solve the following convex optimization problem
to find the optimal $|\rho|$:
\begin{eqnarray}
\begin{gathered}
\mathop {\max }\limits_{\left| \rho  \right|} \det (\mathbf{I} \mathbf{ + HQ}_r \mathbf{H}^\mathbf{\dag } ) \hfill \\
  \mbox{ subject to } 0 \leq \left| \rho  \right| \leq 1 \nonumber
\end{gathered}
\end{eqnarray}
where its Lagrangian function is given by
\begin{eqnarray}
L(\left| \rho  \right|,\lambda ) = -\det (\mathbf{I} \mathbf{ +
HQ}_r \mathbf{H}^{\dag } ) + \lambda _1 (\left| \rho \right| - 1) +
\lambda _2 ( - \left| \rho  \right|). \hfill \nonumber
\end{eqnarray}
The Karush-Kuhn-Tucker (KKT) condition is given by
\begin{eqnarray}
\nabla L(\left| \rho \right|,\lambda ) = 2\left| \rho \right|\sqrt
{P_{r1} P_{r2} } \det (\mathbf{H}^{\dag } \mathbf{H}) - 2
|{\mathbf{h}_1^{\dag } \mathbf{h}_2 } | + \lambda _1 - \lambda _2  =
0. \nonumber
\end{eqnarray}
Solving this for $|\rho|$ gives
\begin{eqnarray}
\left| \rho  \right| = \frac{{| {\mathbf{h}_1^\mathbf{\dag }
\mathbf{h}_2 } | + (\lambda _2  - \lambda _1 )/2}} {{\sqrt {P_{r1}
P_{r2} } \det (\mathbf{H}^{\dag } \mathbf{H})}}. \nonumber
\end{eqnarray}
From complementary slackness, it must be satisfied that $\lambda _1
(|\rho|-1)=0$ and $\lambda _2 |\rho|=0$. Thus, if $0 < \left|
\rho \right| < 1$, then the optimal solution would be $\lambda _1 =
0$, $\lambda _2 = 0$, and
\begin{eqnarray}
\left| \rho  \right| = \frac{{| {\mathbf{h}_1^\mathbf{\dag }
\mathbf{h}_2} |}} {{\sqrt {P_{r1} P_{r2} } \det (\mathbf{H}^{\dag }
\mathbf{H})}}. \nonumber
\end{eqnarray}
Likewise, we also have the following two sets of solutions,
\begin{eqnarray}
  \left| \rho  \right| = 0,\ \lambda _1  =  0,\ \lambda _2  =  2| {\mathbf{h}_1^{\dag } \mathbf{h}_2 } |
  \nonumber
\end{eqnarray}
\begin{eqnarray}
  \left| \rho  \right| = 1,\ \lambda _1  =  2\sqrt {P_{r1} P_{r2} } \det (\mathbf{H}^{\dag } \mathbf{H})-2| {\mathbf{h}_1^{\dag } \mathbf{h}_2} |,\ \lambda _2  =
  0.
  \nonumber
\end{eqnarray}
In other words, the function $\log \det (\mathbf{I} +\mathbf{HQ}_r
\mathbf{H}^{\dag } )$ is a quadratic and concave function of
$|\rho|$ with its maximum at $|\rho|' = \frac{{|{\mathbf{h}_1^{\dag
} \mathbf{h}_2}|}} {{\sqrt {P_{r1} P_{r2} } \det (\mathbf{H}^{\dag }
\mathbf{H})}} \geq 0$ without constraints. If $|\rho|' \leq 1$, the
constraint is inactive so $|\rho|'$ maximizes the objective
function. If $|\rho|'> 1$, it violates the constraint, and the
objective function has its maximum at the boundary of the feasible
set $|\rho|=1$.


\newpage

\begin{figure}[t!]
  \begin{center}
  \leavevmode \epsfxsize=0.8\textwidth
  \epsffile{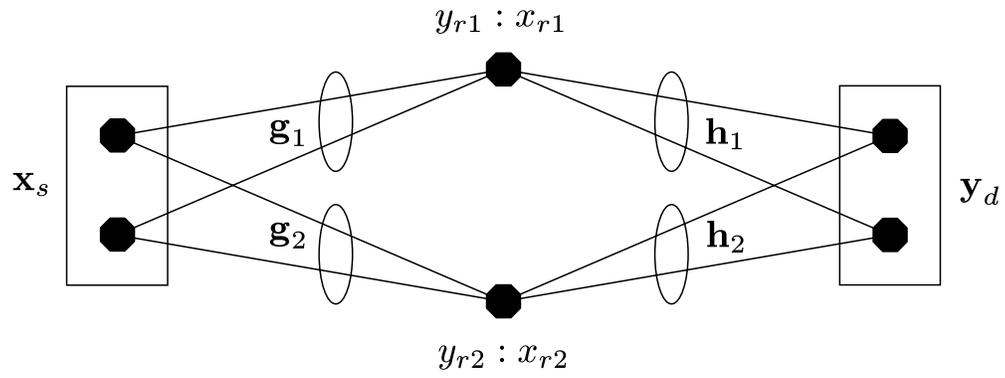}
  \caption{System model}
  \label{f:systemmodel}
  \end{center}
\end{figure}

\begin{figure}[t!]
  \begin{center}
  \leavevmode \epsfxsize=0.8\textwidth
  \epsffile{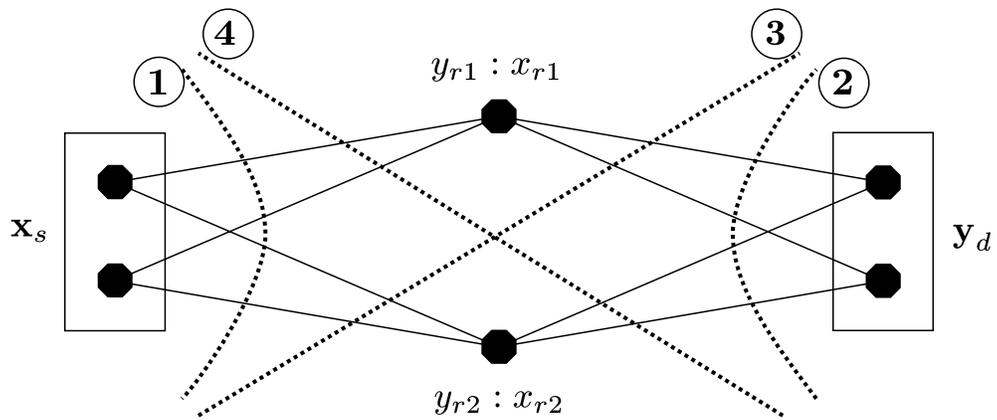}
  \caption{Cut-set upper bounds}
  \label{f:cut-set}
  \end{center}
\end{figure}

\newpage

\begin{figure}[tp]
  \begin{center}
    \mbox{
      \subfigure[Optimal correlation versus $\varphi_h$]{\includegraphics[width=0.5\textwidth]{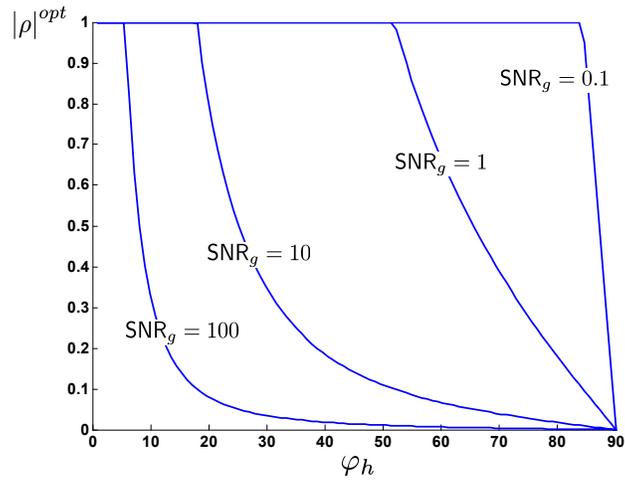}}
      }
    \mbox{
      \subfigure[Optimal correlation versus $\textsf{SNR}_g$]{\includegraphics[width=0.5\textwidth]{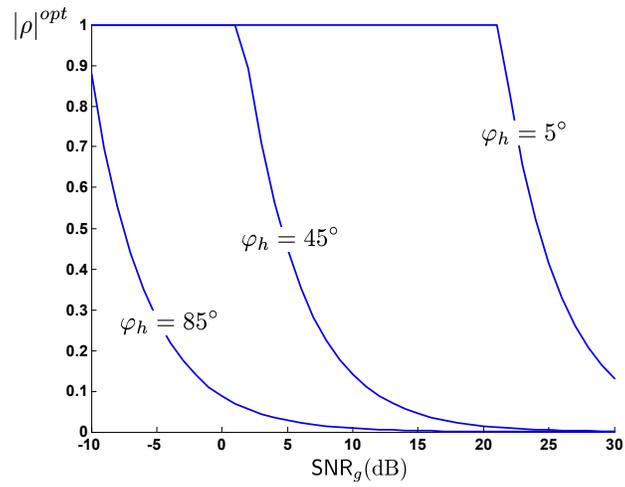}}
      }
    \caption{Optimal correlation as a function of $\varphi_h$ and $\textsf{SNR}_g$}
    \label{fig:rho}
  \end{center}
\end{figure}

\newpage

\begin{figure}[tp]
  \begin{center}
    \mbox{
      \subfigure[$\varphi_h=0^\circ$, $|\rho|^{opt}=1$.]{\includegraphics[width=0.5\textwidth]{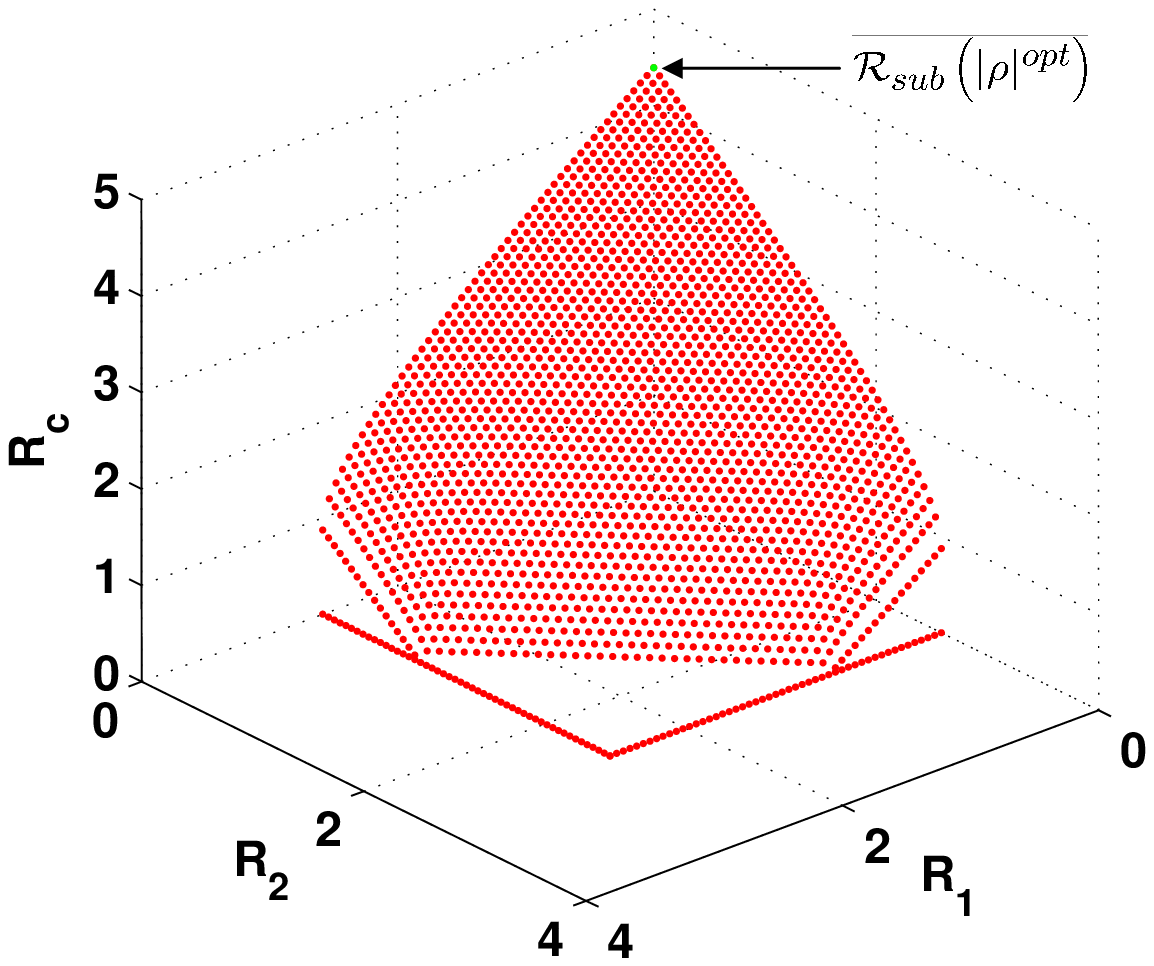}}
      \subfigure[$\varphi_h=25^\circ$, $|\rho|^{opt}=1$.]{\includegraphics[width=0.5\textwidth]{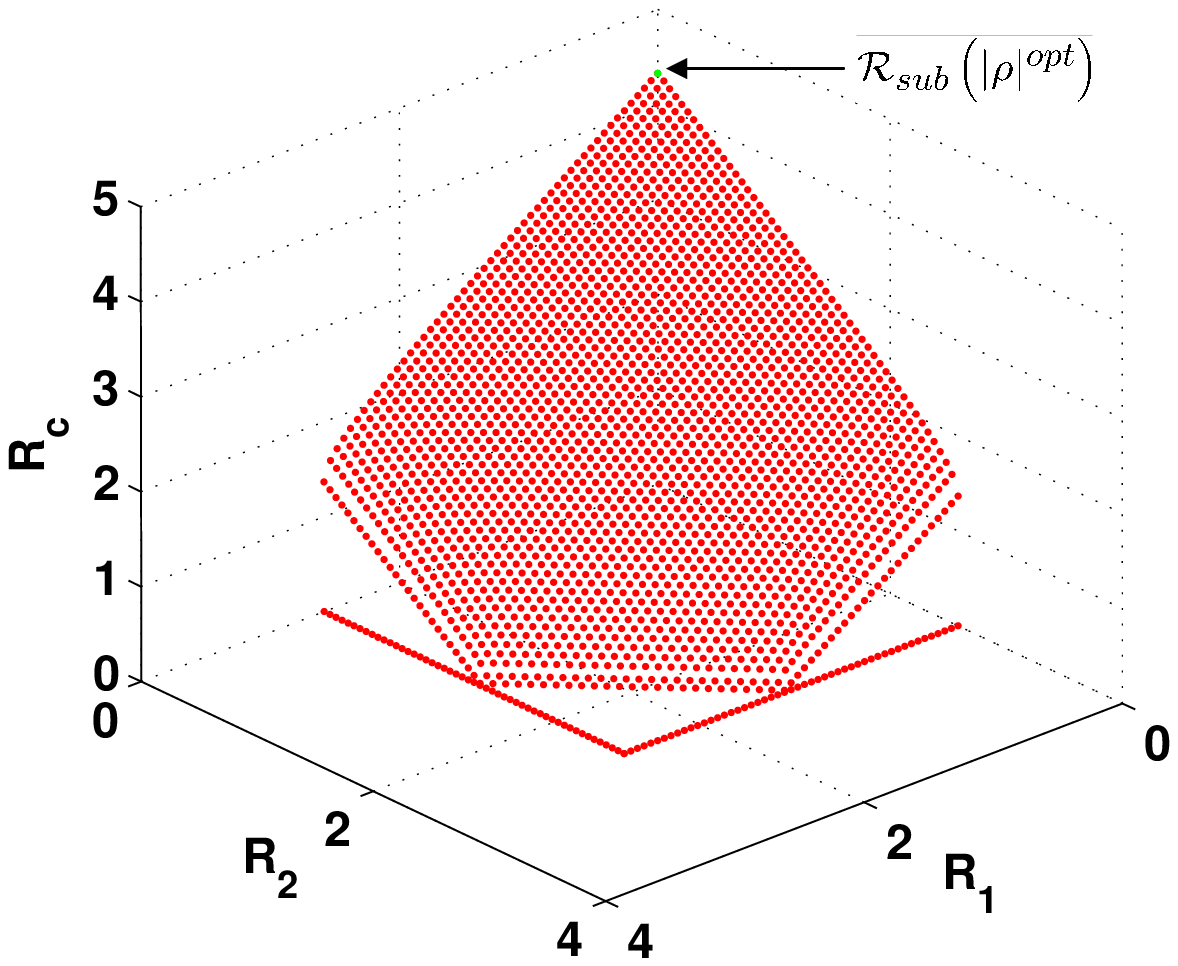}}
      }
    \mbox{
      \subfigure[$\varphi_h=35^\circ$, $|\rho|^{opt}=0.5$.]{\includegraphics[width=0.5\textwidth]{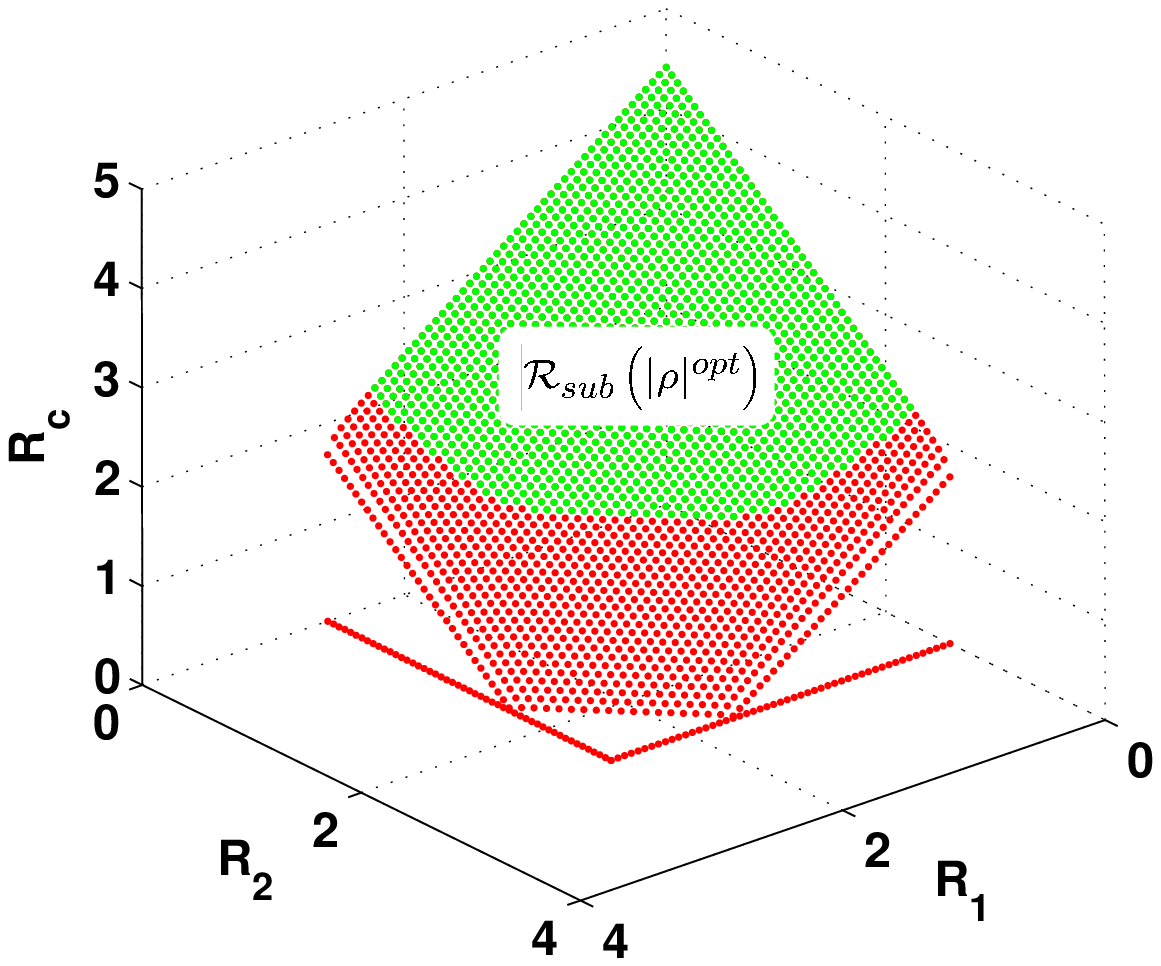}}
      \subfigure[$\varphi_h=90^\circ$, $|\rho|^{opt}=0$.]{\includegraphics[width=0.5\textwidth]{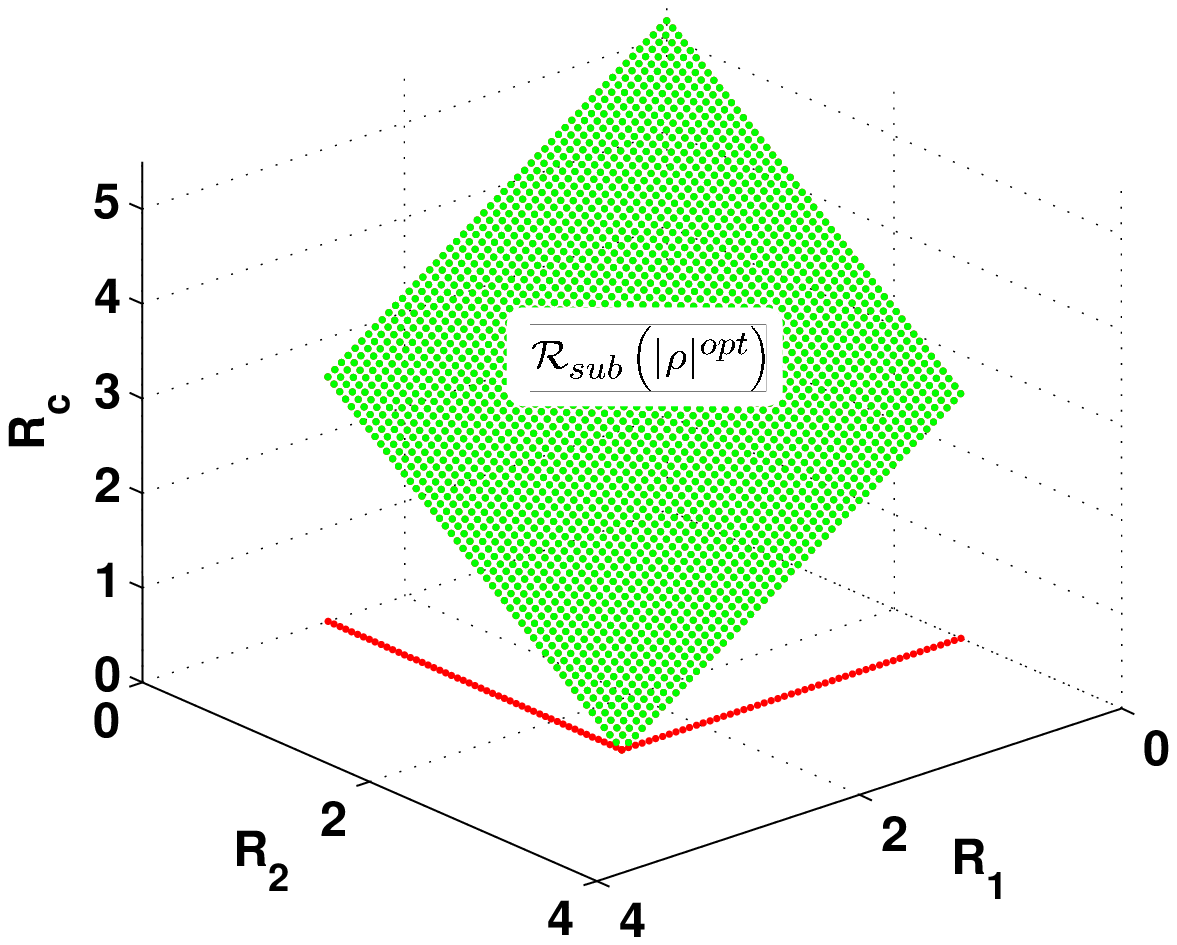}}
      }
    \caption{SIMO MAC achievable regions with $P_{r1}=P_{r2}=5$ and $||\mathbf{h_1}||=||\mathbf{h_2}||=1$}
    \label{fig:MAC_region}
  \end{center}
\end{figure}

\newpage

\begin{figure}[tp]
  \begin{flushleft}
    \mbox{
      \subfigure[$\varphi_h=0^\circ$, $|\rho|^{opt}=1$.]{\includegraphics[width=0.5\textwidth]{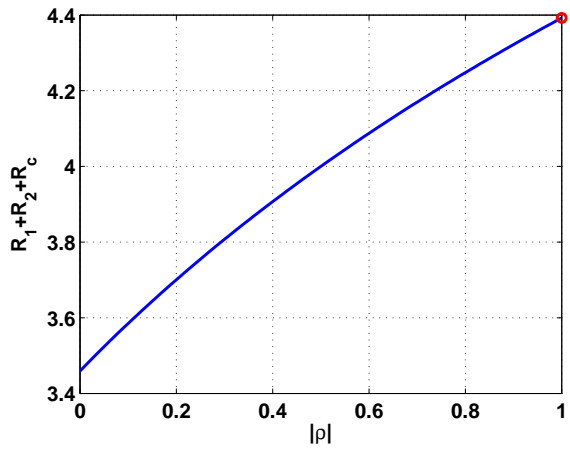}}
      \subfigure[$\varphi_h=25^\circ$, $|\rho|^{opt}=1$.]{\includegraphics[width=0.5\textwidth]{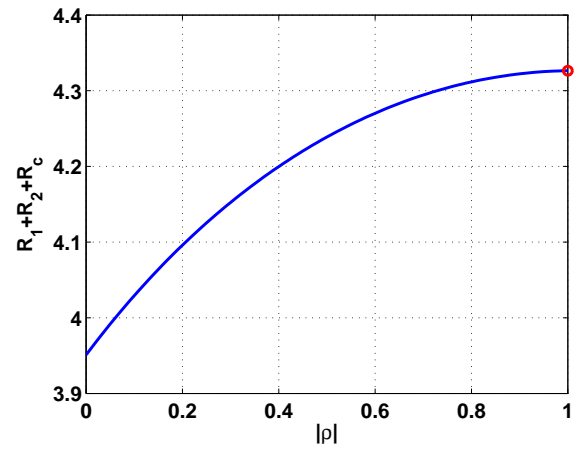}}
      }
    \mbox{
      \subfigure[$\varphi_h=35^\circ$, $|\rho|^{opt}=0.5$.]{\includegraphics[width=0.5\textwidth]{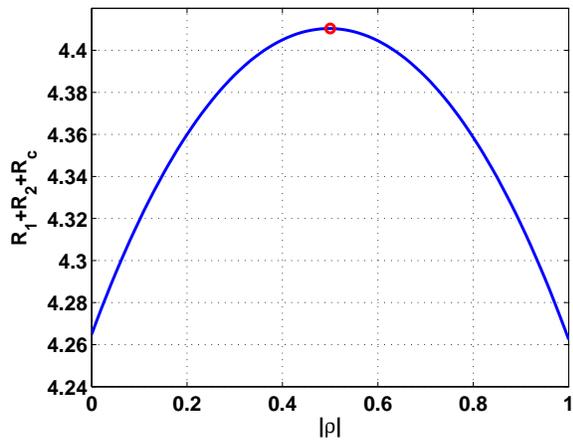}}
      \subfigure[$\varphi_h=90^\circ$, $|\rho|^{opt}=0$.]{\includegraphics[width=0.5\textwidth]{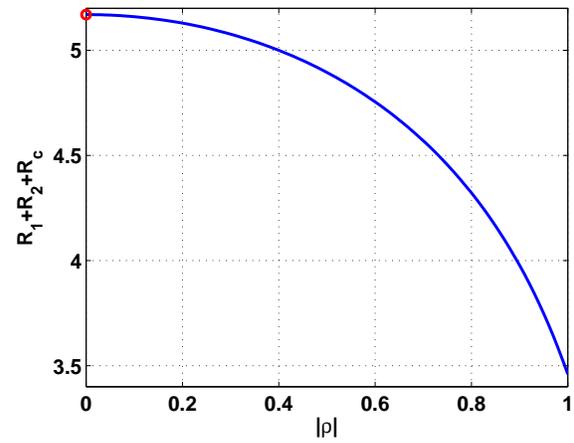}}
      }
    \caption{Sum-rate versus $|\rho|$ with $P_{r1}=P_{r2}=5$ and $||\mathbf{h_1}||=||\mathbf{h_2}||=1$}
    \label{fig:Rsum}
  \end{flushleft}
\end{figure}

\newpage

\begin{figure}[tp]
  \begin{center}
    \mbox{
      \subfigure[Parallel channel vectors: $\varphi_g = 0^\circ$]{\includegraphics[width=0.5\textwidth]{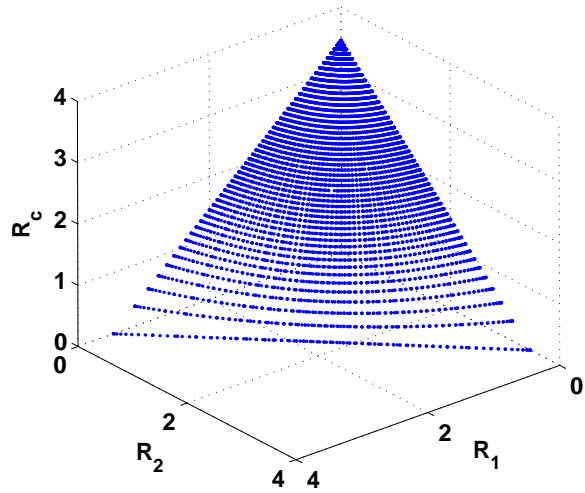}}
      }
    \mbox{
      \subfigure[Orthogonal channel vectors: $\varphi_g = 90^\circ$]{\includegraphics[width=0.5\textwidth]{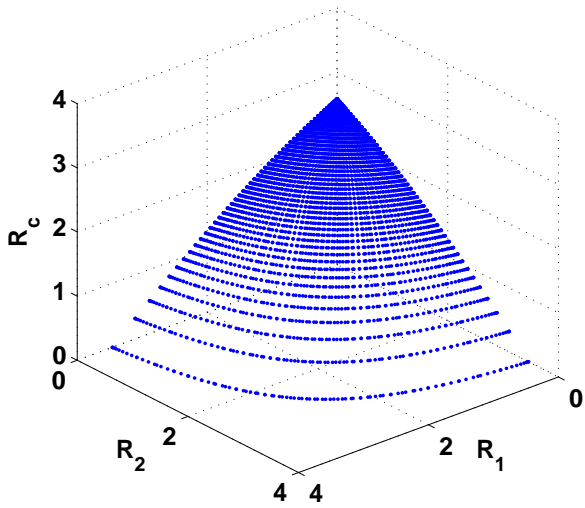}}
      }
    \caption{MISO BC achievable regions with $P_s=10$ and $||\mathbf {g}_1||=||\mathbf {g}_2||=1$}
    \label{fig:BC_region}
  \end{center}
\end{figure}

\newpage

\begin{figure}[tp]
  \begin{center}
    \mbox{
      \subfigure[Three-dimensional region]{\includegraphics[width=0.5\textwidth]{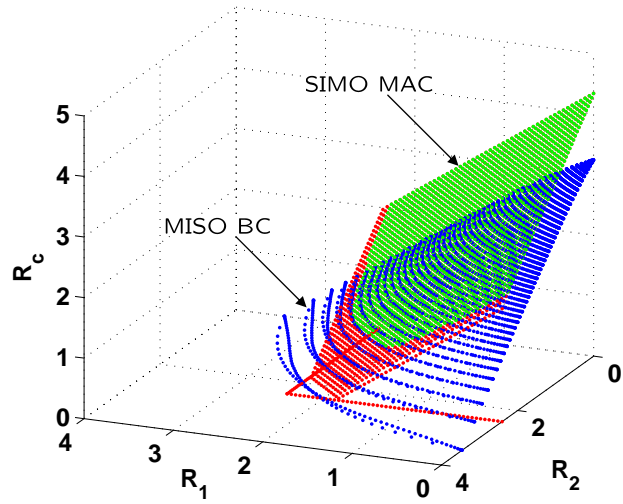}}
      }
    \mbox{
      \subfigure[Cross section along the surface $R_1=R_2$]{\includegraphics[width=0.5\textwidth]{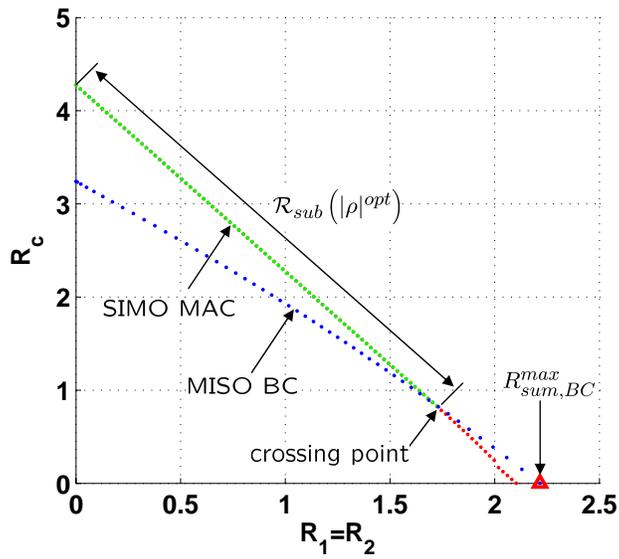}}
      }
    \caption{DF achievable region: $P_{\mathbf s}=10$, $P_{r1} = P_{r2} = 4.17$, $\varphi_g =\varphi_h = 46.3942^ \circ$, $\left| \rho  \right|^{opt}  = 0.31$ and $R_c^{th}=0.79$.}
    \label{fig:BCMAC}
  \end{center}
\end{figure}

\newpage

\begin{figure}[tp]
  \begin{flushleft}
    \mbox{
      \subfigure[Multiple DF optimal points.]{\includegraphics[width=0.5\textwidth]{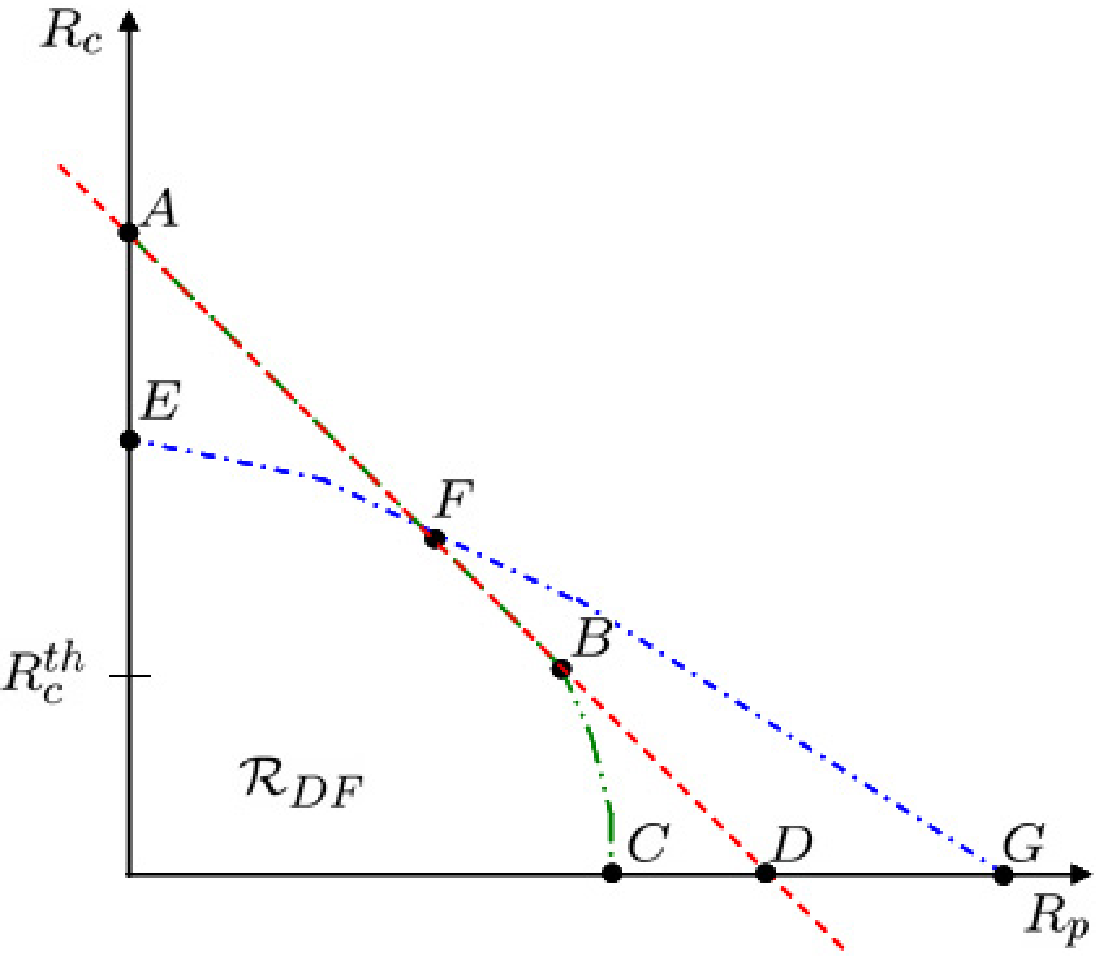}}
      \subfigure[Single DF optimal point.]{\includegraphics[width=0.5\textwidth]{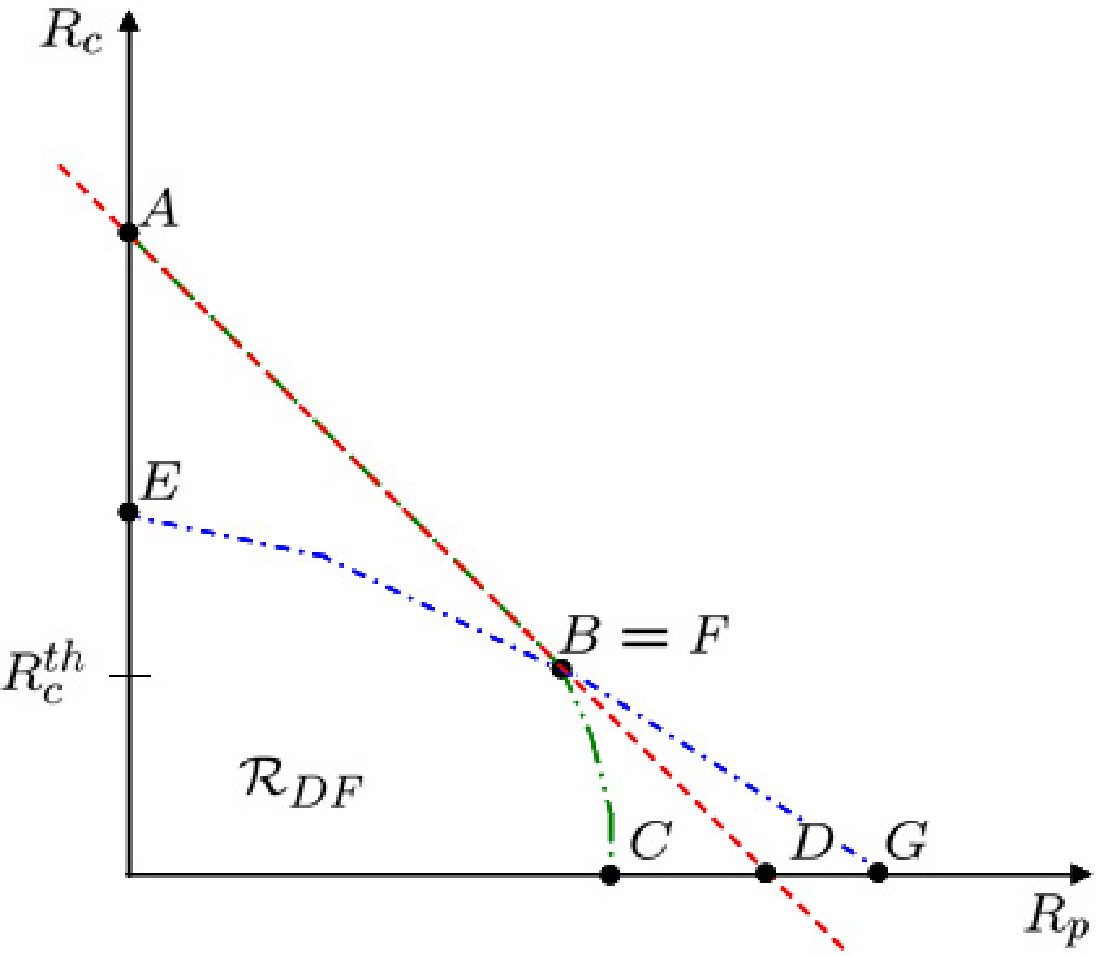}}
      }
    \mbox{
      \subfigure[No DF optimal point.]{\includegraphics[width=0.5\textwidth]{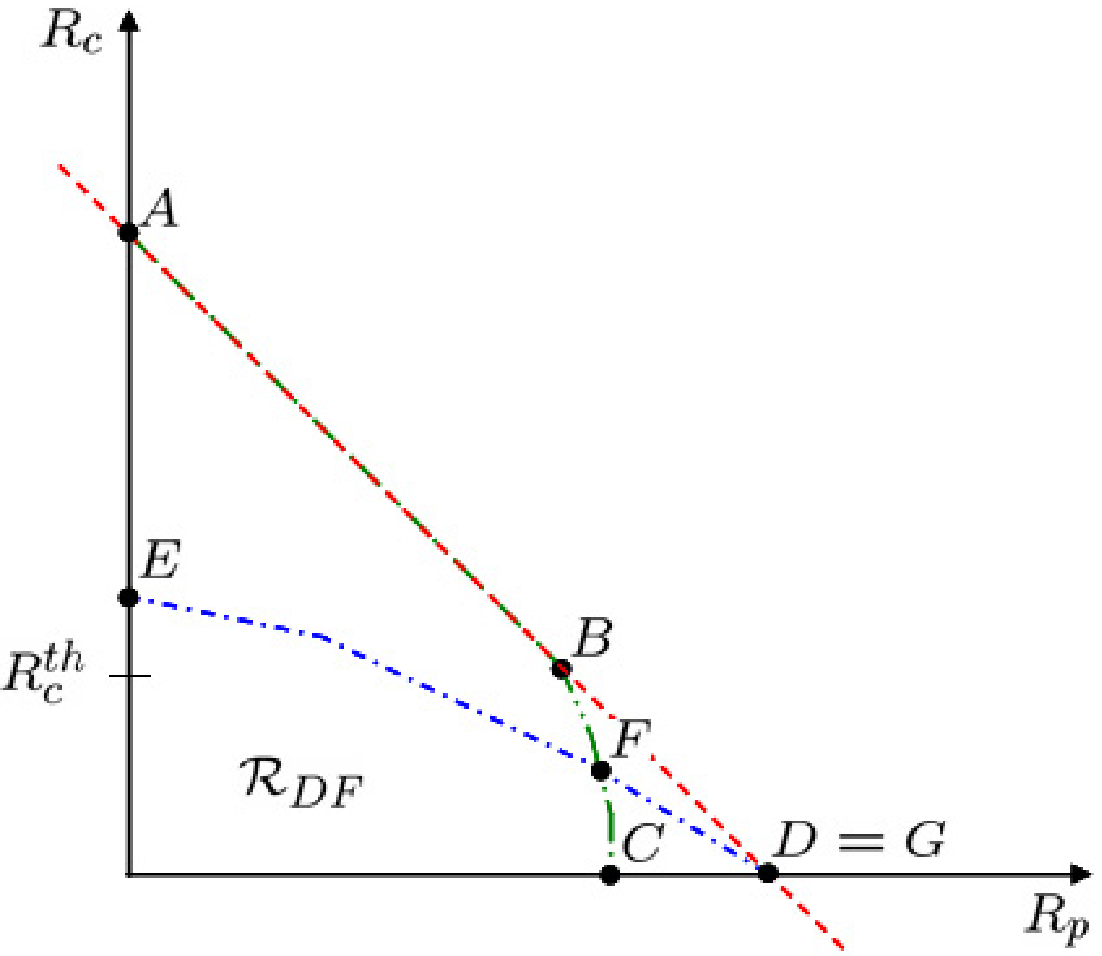}}
      \subfigure[No DF optimal point.]{\includegraphics[width=0.5\textwidth]{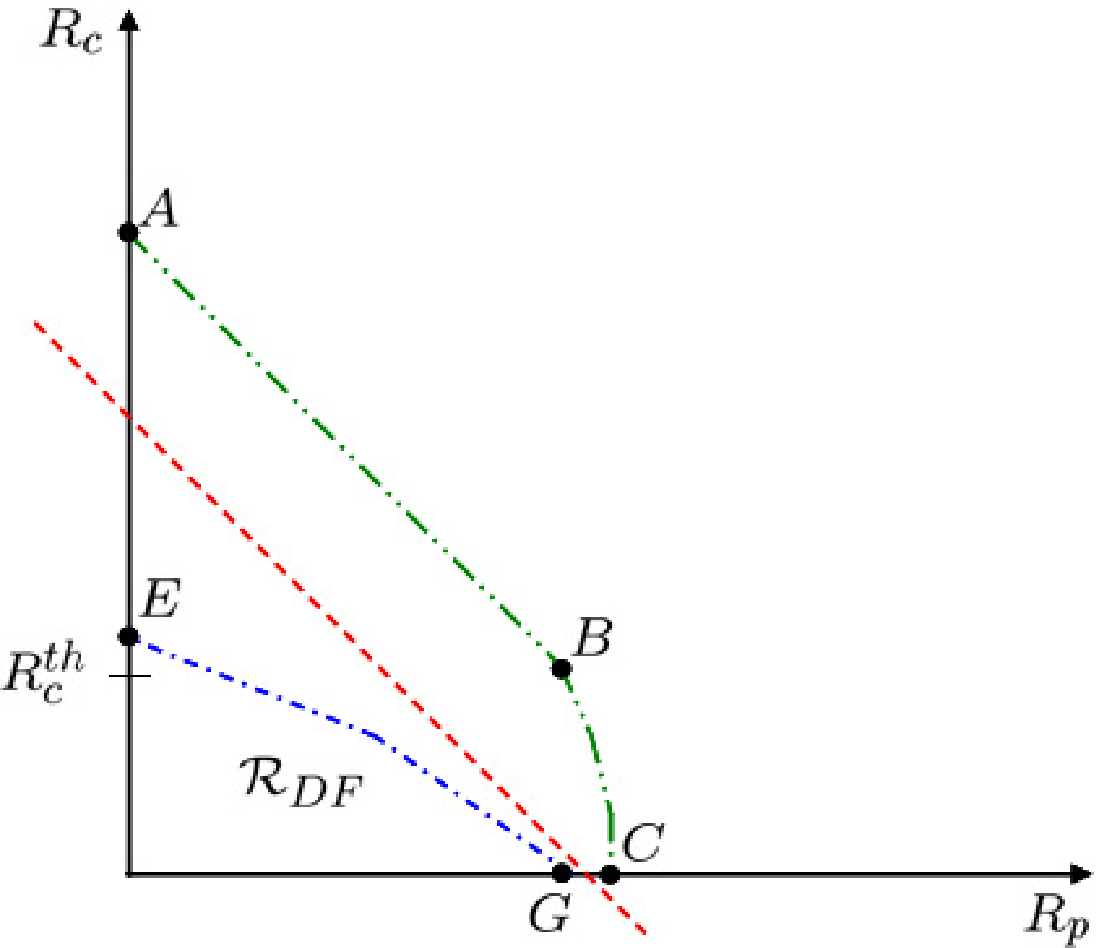}}
      }
    \caption{Symmetric channels.}
    \label{fig:Symmetric}
  \end{flushleft}
\end{figure}

\newpage

\begin{figure}[tp]
  \begin{center}
    \mbox{
      \subfigure[$\left\|{\bf g}_1\right\|=1$, $\left\|{\bf g}_2\right\|=0.16$ and $b^{peak}=1.67$]{\includegraphics[width=0.5\textwidth]{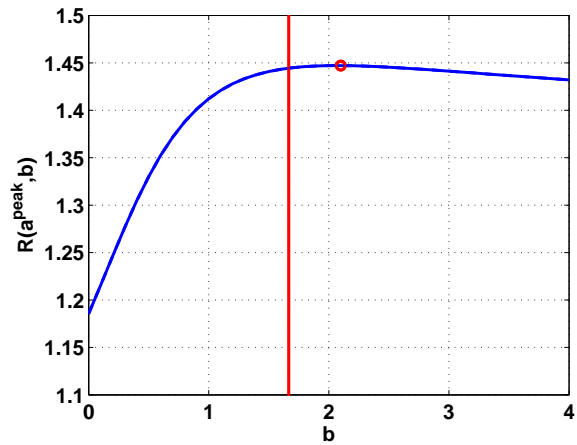}}
      }
    \mbox{
      \subfigure[$\left\|{\bf g}_1\right\|=1$, $\left\|{\bf g}_2\right\|=0.09$ and $b^{peak}=1.85$]{\includegraphics[width=0.5\textwidth]{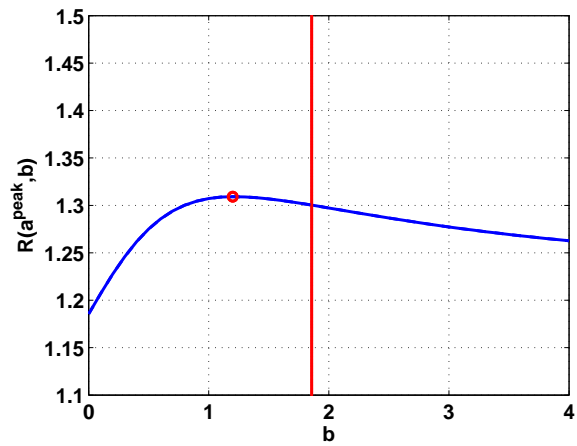}}
      }
    \caption{$R(a^{peak},b)$ versus $b$ with $P_s = 10$, $P_{r1}  = P_{r2}  = 5$, $\left\|{\bf h}_1\right\|=\left\|{\bf h}_2\right\|=1$ and $\varphi_g=\varphi_h=46.40^\circ$}
    \label{fig:rate_b}
  \end{center}
\end{figure}

\newpage

\begin{figure}[tp]
  \begin{flushleft}
    \mbox{
      \subfigure[$\varphi_g=90^\circ$ and $\varphi_h=90^\circ$.]{\includegraphics[width=0.45\textwidth]{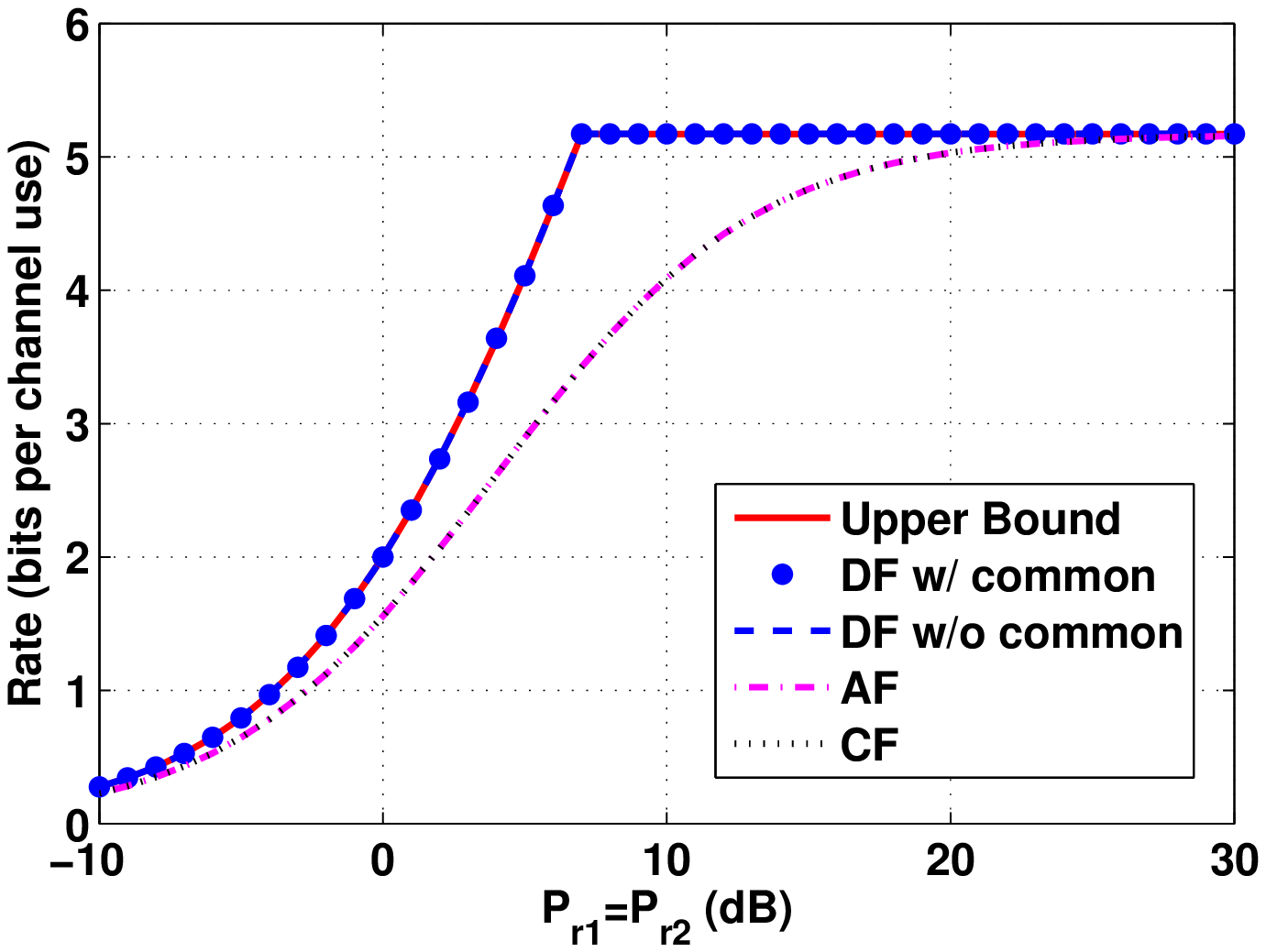}}
      \subfigure[$\varphi_g=90^\circ$ and $\varphi_h=0^\circ$.]{\includegraphics[width=0.45\textwidth]{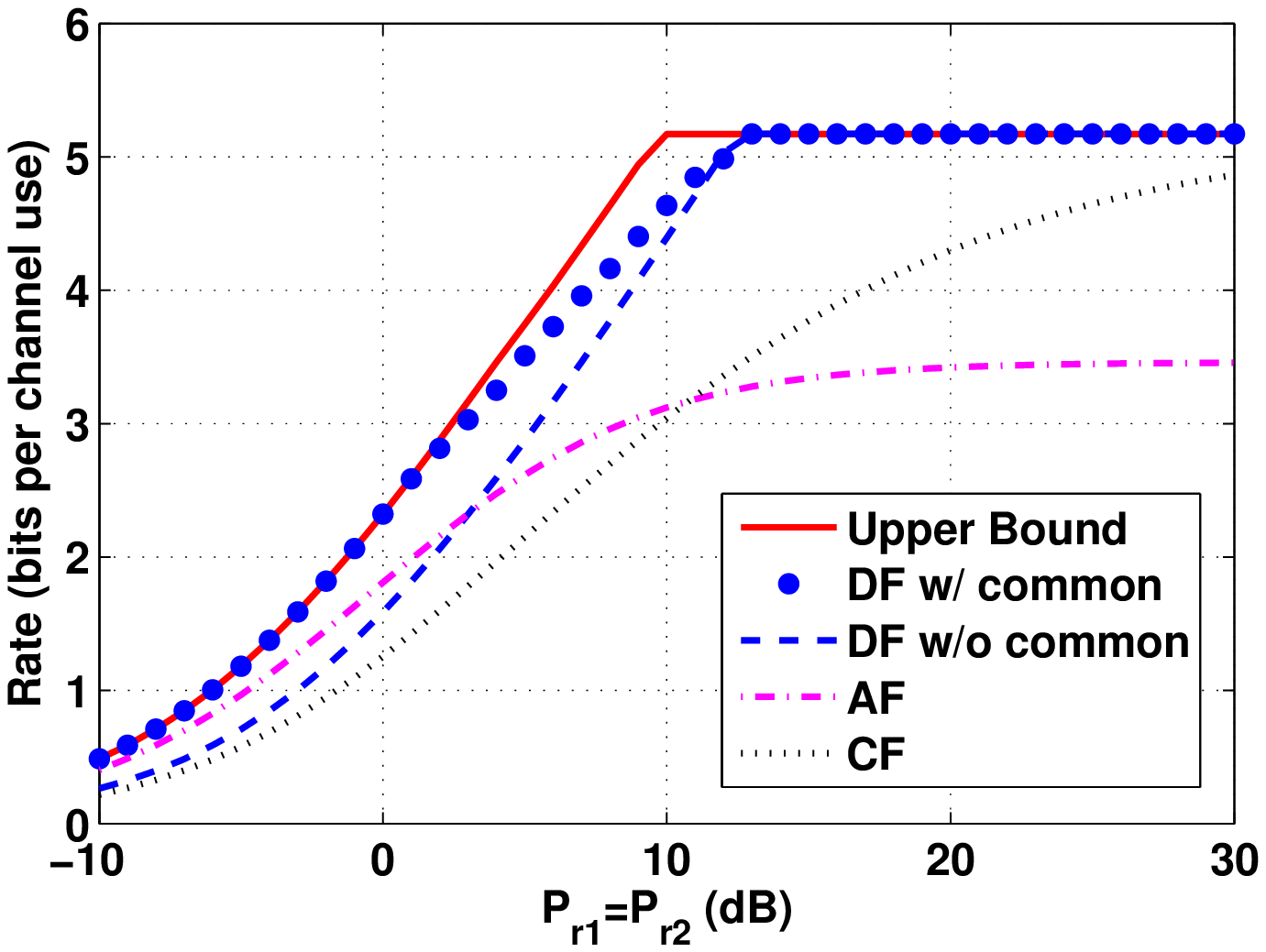}}
      }
    \mbox{
      \subfigure[$\varphi_g=46.3972^\circ$ and $\varphi_h=90^\circ$.]{\includegraphics[width=0.45\textwidth]{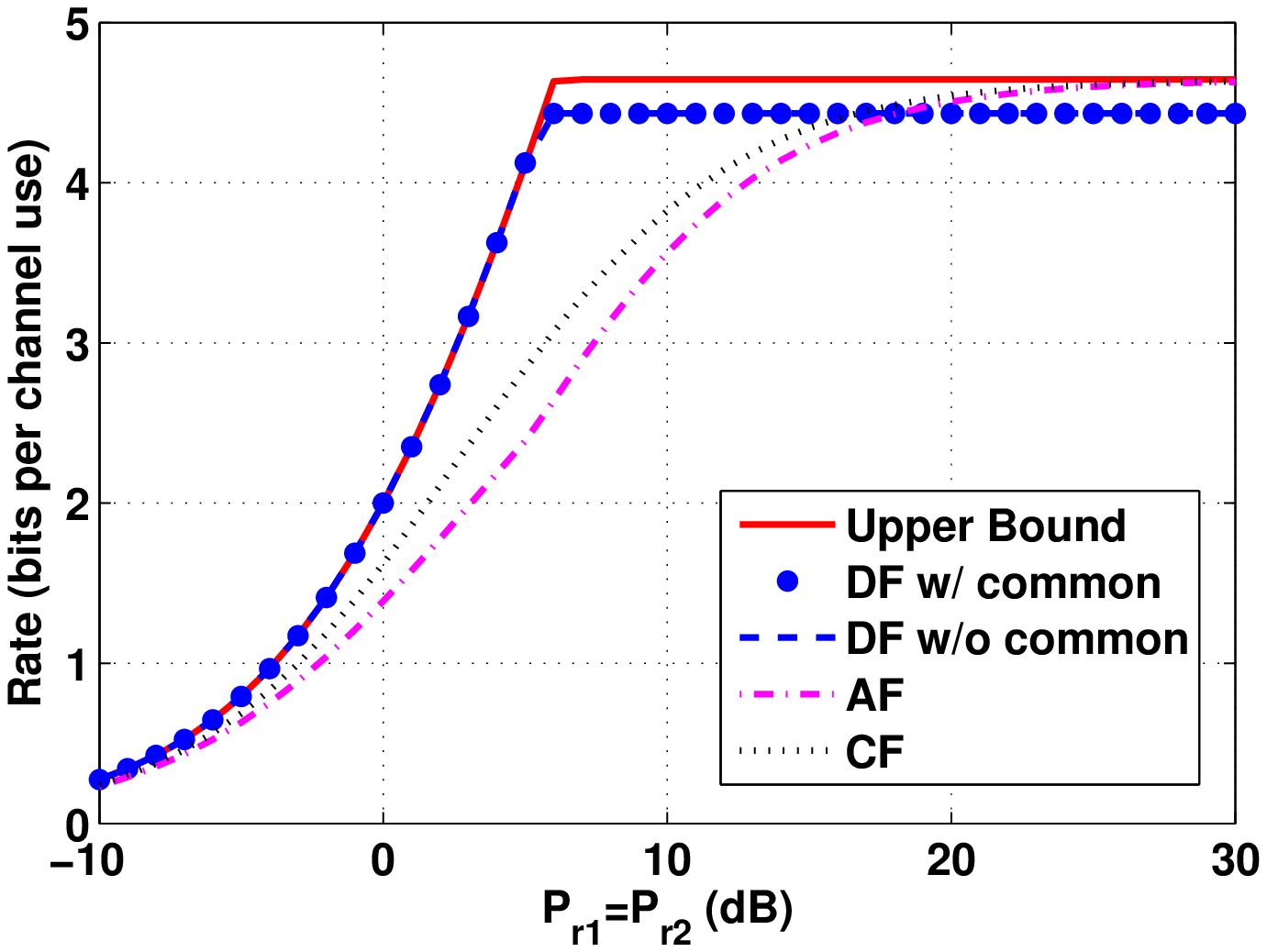}}
      \subfigure[$\varphi_g=46.3972^\circ$ and $\varphi_h=0^\circ$.]{\includegraphics[width=0.45\textwidth]{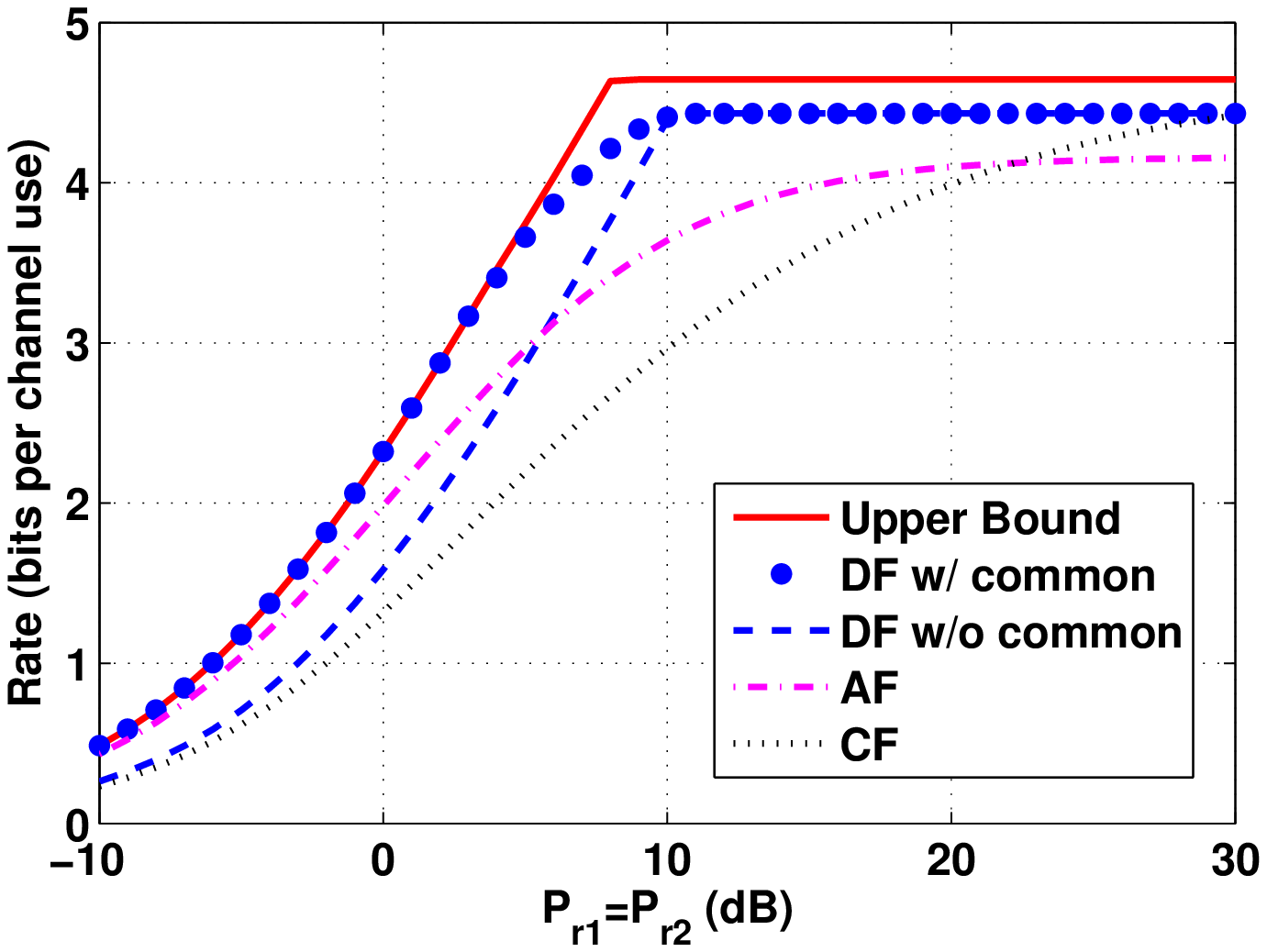}}
      }
    \mbox{
      \subfigure[$\varphi_g=0^\circ$ and $\varphi_h=90^\circ$.]{\includegraphics[width=0.45\textwidth]{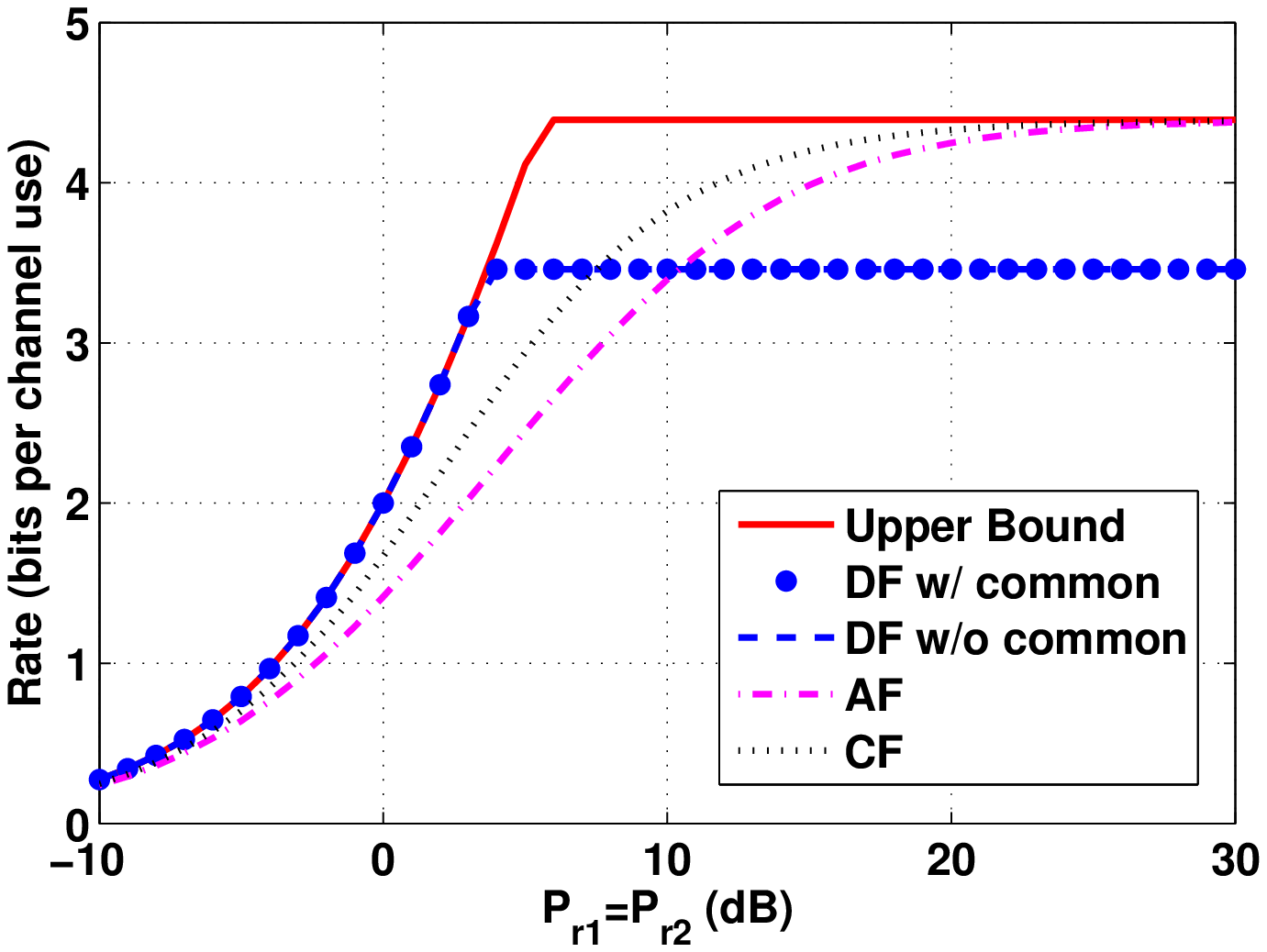}}
      \subfigure[$\varphi_g=0^\circ$ and $\varphi_h=0^\circ$.]{\includegraphics[width=0.45\textwidth]{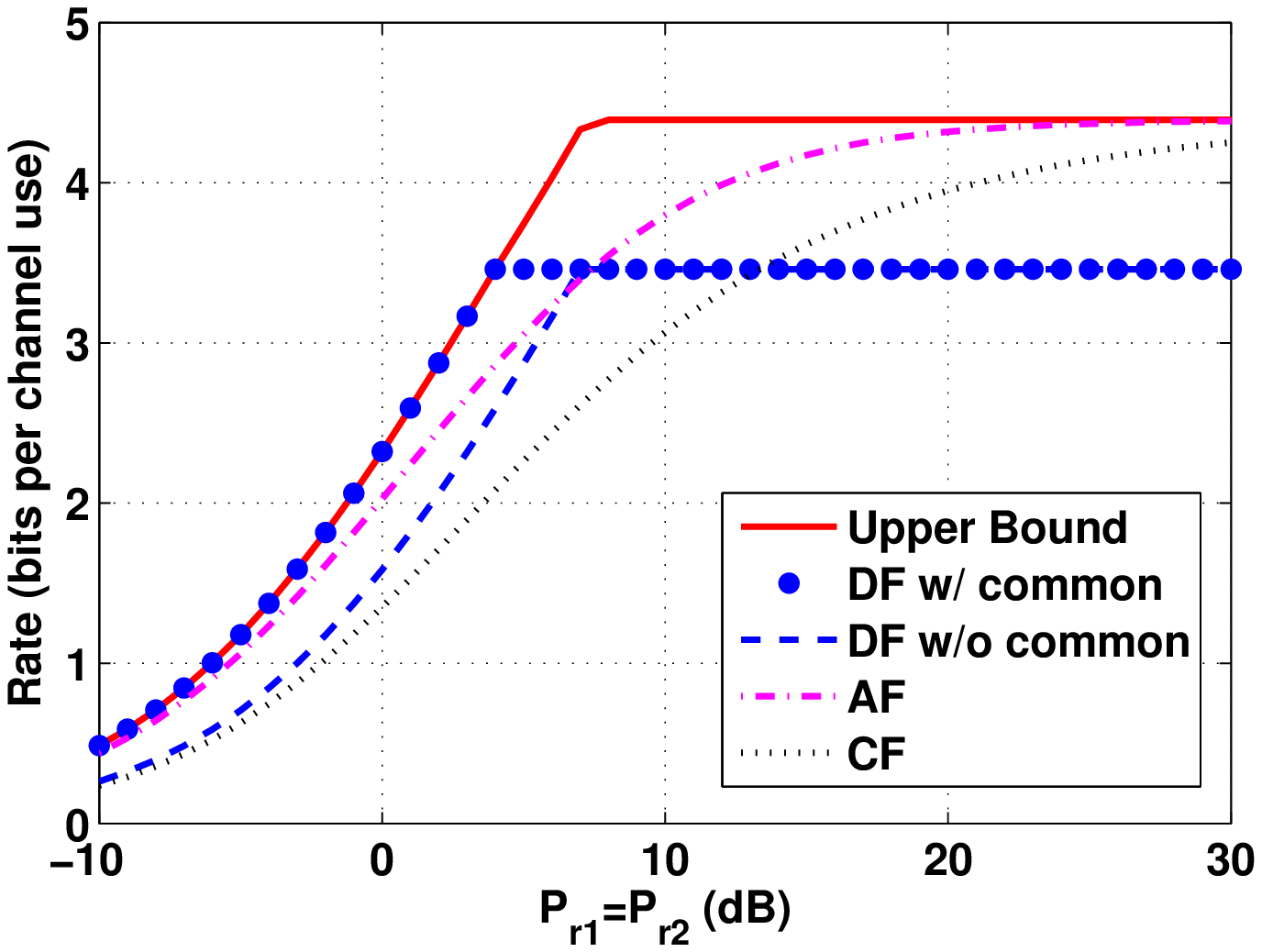}}
      }
    \caption{Upper bound and achievable rates versus relay power. Note that $P_s =10$ (dB) and $\left\|{\bf g}_1\right\|=\left\|{\bf g}_2\right\|=\left\|{\bf h}_1\right\|=\left\|{\bf h}_2\right\|=1$.}
    \label{f:sum-rate}
  \end{flushleft}
\end{figure}

\end{document}